\documentclass[english,11pt]{article}

\usepackage[T1]{fontenc}
\usepackage[utf8]{inputenc}
\usepackage[margin=1in]{geometry}
\usepackage{lmodern}
\usepackage{microtype}
\usepackage{amsmath,amssymb,amsfonts,amsthm}
\usepackage{array}
\usepackage{booktabs}
\usepackage{graphicx}
\usepackage{algorithm}
\usepackage{algpseudocode}
\usepackage{natbib}
\usepackage{tabularx}
\usepackage{xcolor}
\usepackage{bm}
\usepackage[toc,page,header]{appendix}
\usepackage{authblk}
\usepackage[colorlinks=true,linkcolor=blue!50!black,citecolor=blue!50!black,urlcolor=blue!50!black,breaklinks=true]{hyperref}
\newcommand{\R}{\mathbb{R}}
\newcommand{\simplex}{\Delta_n}
\newcommand{\hadpow}[2]{#1^{\circ #2}}

\newtheorem{theorem}{Theorem}
\newtheorem{proposition}[theorem]{Proposition}
\newtheorem{corollary}[theorem]{Corollary}
\newtheorem{lemma}[theorem]{Lemma}
\newtheorem{assumption}[theorem]{Assumption}
\theoremstyle{definition}
\newtheorem{definition}[theorem]{Definition}
\theoremstyle{remark}
\newtheorem{remark}[theorem]{Remark}

\newcommand{\mvskTitle}{Yau's Affine-Normal Descent for Large-Scale Unrestricted Higher-Moment Portfolio Optimization}
\newcommand{\mvskAbstractText}{%
Unrestricted mean-variance-skewness-kurtosis portfolio optimization can capture asymmetry and tail risk, but sample-moment formulations become computationally impractical when the asset universe is large: they produce dense nonconvex quartic objectives with prohibitive coskewness and cokurtosis tensors and anisotropic, ill-conditioned level sets. We develop a structure-exploiting algorithm based on Yau's affine-normal descent that follows affine-normal directions of the current level set while working directly with the return matrix. The method avoids explicit higher-order tensors and exploits the quartic structure for exact sample oracles, derivative evaluation, and exact line search. We also provide theory for the reduced simplex formulation, including regularity and convexity conditions that separate data-map geometry from investor preference coefficients. Computational results show a clear implementation split: a direct configuration is effective on the standard small benchmark, whereas a preconditioned conjugate-gradient configuration with stall recovery becomes the preferred large-scale implementation by the upper end of the hundreds and remains competitive as the asset universe moves into the thousands. On a 5-minute A-share panel with 5,440 stocks, the method makes direct full-universe comparisons with exact mean-variance portfolios feasible and shows on the baseline split that the incremental value of higher moments is strongest at moderate return targets.%
}
\newcommand{\mvskKeywordsText}{Yau's affine-normal descent, portfolio optimization, higher moments, computational finance}
\newcommand{\mvskSubjectClassText}{Portfolio theory; nonlinear programming; large-scale algorithms}

\hypersetup{
  pdftitle={Yau's Affine-Normal Descent for Large-Scale Unrestricted Higher-Moment Portfolio Optimization},
  pdfauthor={Ya-Juan Wang, Yi-Shuai Niu, Artan Sheshmani, Shing-Tung Yau}
}

\title{\mvskTitle}
\author[1]{Ya-Juan Wang\thanks{\texttt{wangyajuan@bimsa.cn}}}
\author[1]{Yi-Shuai Niu\thanks{Corresponding author. \texttt{niuyishuai@bimsa.cn}}}
\author[1,3]{Artan Sheshmani\thanks{\texttt{artan@mit.edu}}}
\author[1,2]{Shing-Tung Yau\thanks{Corresponding author. \texttt{styau@tsinghua.edu.cn}}}
\affil[1]{Beijing Institute of Mathematical Sciences and Applications (BIMSA), Beijing, China}
\affil[2]{Yau Mathematical Sciences Center, Tsinghua University, Beijing, China}
\affil[3]{IAIFI Institute, Massachusetts Institute of Technology, Cambridge, MA 02139, USA}
\date{}

\begin{document}

\maketitle

\begin{abstract}
\mvskAbstractText
\end{abstract}


\medskip
\noindent\textbf{Subject classifications.} \mvskSubjectClassText

\bigskip
\tableofcontents
\bigskip

\section{Introduction}
\label{sec:introduction}
\enlargethispage{\baselineskip}

Mean--variance analysis remains the standard benchmark for portfolio design, but many investors care about asymmetry and downside tail exposure as well as variance \citep{KrausLitzenberger1976,ScottHorvath1980}. This motivates higher-moment extensions such as mean--variance--skewness--kurtosis portfolio optimization \citep{Markowitz1952,deAthaydeFlores2004,JondeauRockinger2006,NiuWangThiPham2019,ZhouPalomar2021,Palomar2025}. The practical question, however, is not whether higher moments are conceptually attractive. It is whether the unrestricted sample-moment model can be solved at a scale large enough to compare it credibly with exact mean--variance portfolios and to determine when the extra modeling flexibility matters economically.

That question has remained difficult mainly for computational reasons. The unrestricted model is dense, nonconvex, and expensive to evaluate through its traditional tensor representation. Explicit covariance, coskewness, and cokurtosis objects become costly to build, store, and differentiate as the asset universe grows. As a result, much of the empirical literature remains concentrated on modest dimensions \citep{deAthaydeFlores2004,JondeauRockinger2006,ZhouPalomar2021}. This makes it hard to tell whether higher moments continue to affect portfolio choice after diversification or whether the reported gains are mainly a small-universe artifact.

This paper addresses that bottleneck with a dedicated YAND-MVSK algorithm based on Yau's affine-normal descent (YAND). The YAND terminology is used here in its algorithmic sense. \citet{ChengChengYau2005} introduced the affine-normal direction as a geometric search direction for minimization. The recent YAND framework of \citet{NiuSheshmaniYau2026YAND} turns that idea into a line-search method for smooth unconstrained optimization, including possibly nonconvex objectives: the search direction is defined by the affine normal of the current level set, adapts to anisotropic curvature, is robust to affine scaling, and comes with convergence guarantees ranging from global stationarity under standard smoothness conditions to linear rates under strong convexity or PL-type structure and local quadratic convergence near nondegenerate minimizers. A companion computational development shows that affine-normal evaluation can be made matrix-free through log-determinant geometry, tangent linear solves, Hessian-vector products, and structured derivative kernels rather than explicit high-order tensor contraction \citep{NiuSheshmaniYau2026LogDet}.

This viewpoint is especially relevant for unrestricted MVSK: its sample objective is a linear return map followed by a separable quartic response, and much of its numerical difficulty appears as anisotropic scaling and elongated level-set geometry. Our contribution is to specialize the general YAND framework to the simplex-constrained sample-moment MVSK problem. The investor-facing model is unchanged, but the numerical interface is rebuilt around the return matrix. The method combines exact sample-oracle calculations, reduced-coordinate affine-normal steps on the simplex, exact quartic line search, and boundary-aware continuation, thereby avoiding explicit higher-order tensors while preserving the unrestricted sample-moment objective.

The paper also explains when this interface should work well. We show that the simplex-constrained problem admits an exact reduced-coordinate representation for YAND, derive regularity bounds separating data geometry from preference-driven curvature, and give a coefficient-only convexity test covering the standard CRRA calibration.

The empirical evidence is organized to answer three linked questions. Can the method handle controlled geometric deterioration inside a convex regime? How should it be configured as the asset dimension grows on standard unrestricted-MVSK benchmarks? Once unrestricted MVSK becomes feasible at full-universe scale, when does it materially change the comparison with exact mean--variance portfolios?

Three findings stand out. First, the exact sample-oracle interface makes unrestricted MVSK operational in dimensions that are out of reach for explicit-tensor implementations. Second, the computational study reveals a clear implementation split: the direct configuration remains effective on the common small benchmark, whereas the PCG-based large-scale configuration with stall recovery becomes the preferred implementation by the upper end of the hundreds and remains competitive as the asset universe moves into the thousands. Third, the baseline real-data split shows that the value of higher moments is conditional rather than uniform. Kurtosis-aware MVSK adds the most at moderate return floors on that split, while appendix rolling-window checks show that the separation from exact MV is market-window dependent rather than universal.

The remainder of the paper proceeds as follows. Section~\ref{sec:problem-setting} formulates the unrestricted sample-moment MVSK problem and explains why scale matters for the financial comparison with mean--variance portfolios. Section~\ref{sec:literature-positioning} positions the paper relative to the main strands of the higher-moment literature. Section~\ref{sec:yand-algorithm} presents the YAND-MVSK solver. Section~\ref{sec:theory} develops the theoretical results. Section~\ref{sec:computational-evidence} reports the computational study, and Section~\ref{sec:when-higher-moments-matter} turns to the large-scale economic comparison with exact mean--variance. Section~\ref{sec:conclusion} concludes.

\section{Problem Setting and the Role of Scale}
\label{sec:problem-setting}

Let $x \in \simplex := \{x \in \R^n_+ : \bm{1}^{\top}x = 1\}$ denote the portfolio weights. In the unrestricted sample-moment MVSK model, the objective is
\begin{equation}
\label{eq:mvsk}
\min_{x \in \simplex} f(x), \qquad
f(x)= -c_1 m_1(x) + c_2 m_2(x) - c_3 m_3(x) + c_4 m_4(x),
\end{equation}
We refer to \eqref{eq:mvsk} as the unrestricted sample-moment MVSK problem, where $c_1,c_2,c_3,c_4 \ge 0$ encode the investor's trade-off across mean, variance, skewness, and kurtosis.

Let $R \in \R^{T \times n}$ be the sample-return matrix, let $r_t^\top$ denote its $t$th row, and let $\mu:=T^{-1}R^\top \bm{1}$ be the sample mean vector. We define
\[
m_1(x):=\mu^\top x, \qquad
m_p(x):=\frac{1}{T}\sum_{t=1}^T \bigl(r_t^\top x-\mu^\top x\bigr)^p, \quad p=2,3,4.
\]
Writing $A:=R-\bm{1}\mu^\top$ and letting $a_t^\top$ denote the $t$th row of $A$, the centered moments can also be expressed through the explicit comoment tensors
\[
\Sigma:=\frac{1}{T}A^\top A \in \R^{n\times n}, \qquad
\mathcal S:=\frac{1}{T}\sum_{t=1}^T a_t^{\otimes 3}, \qquad
\mathcal K:=\frac{1}{T}\sum_{t=1}^T a_t^{\otimes 4},
\]
so that
\[
m_2(x)=x^\top \Sigma x,\qquad
m_3(x)=\langle \mathcal S, x^{\otimes 3}\rangle,\qquad
m_4(x)=\langle \mathcal K, x^{\otimes 4}\rangle.
\]
In this form, unrestricted sample-moment MVSK is visibly a dense quartic polynomial optimization problem on the simplex \citep{deAthaydeFlores2004,JondeauRockinger2006,NiuWangThiPham2019,ZhouPalomar2021}. The model is appealing because it stays close to the data and does not force a factor restriction or a parametric return family. The same generality is also what makes the problem difficult.

It is also useful to record the objective in lifted form. If
\[
\psi(s):=c_2 s^2-c_3 s^3+c_4 s^4,
\qquad
\Phi(z):=\frac{1}{T}\sum_{t=1}^T \psi(z_t),
\]
then
\[
f(x)=-c_1\mu^\top x+\Phi(Ax).
\]
Thus unrestricted MVSK is a linear embedding $x\mapsto Ax$ followed by a separable quartic nonlinearity. The matrix $A$ drives most of the Euclidean anisotropy and scaling distortion, whereas $\psi$ drives the intrinsic nonconvexity. This is a natural setting for the affine-normal viewpoint: YAND is built to respond to level-set geometry rather than to coordinate distortion, and recent YAND theory shows that its behavior is robust under anisotropic affine reparameterizations \citep{NiuSheshmaniYau2026YAND}.

The computational bottleneck is not only nonconvexity but also the explicit tensor interface. In unrestricted instances these objects are typically dense, so forming them from the return matrix requires $O(Tn^2)$, $O(Tn^3)$, and $O(Tn^4)$ arithmetic for $\Sigma$, $\mathcal S$, and $\mathcal K$, while storing them requires $\Theta(n^2)$, $\Theta(n^3)$, and $\Theta(n^4)$ coefficients. The fourth-order tensor alone already has $n^4$ entries: for $n=200$ this is $1.6\times 10^9$ coefficients, and for $n=500$ it is $6.25\times 10^{10}$. Once these tensors are built, objective, gradient, and Hessian calculations are dominated by dense cubic and quartic contractions. Existing unrestricted-MVSK studies therefore either focus on carefully designed local algorithms for the explicit quartic model or adopt return structures that simplify the moment expressions \citep{BirgeChavezBedoya2016,NiuWangThiPham2019,ZhouPalomar2021,WangEtAl2023}.

That diagnosis changes the paper's substantive target. If unrestricted MVSK can be solved only for very small universes, it is hard to tell whether higher moments remain valuable after diversification or whether an MV approximation is already sufficient. If, by contrast, unrestricted MVSK can be solved for hundreds or thousands of assets, one can study when skewness and kurtosis materially change portfolio choice, when tail risk improves out of sample, and when the added modeling complexity is worth paying. The managerial question is therefore inseparable from the computational one.

For this reason the paper uses a scale-separated empirical design. Moderate dimensions are used for head-to-head comparison with established unrestricted-MVSK baselines, because that is where such comparisons remain fair. Larger dimensions are used both for configuration diagnosis and for direct MV-versus-MVSK evaluation in more realistic universes. The paper is organized around this distinction.

\section{Related Literature}
\label{sec:literature-positioning}

The higher-moment portfolio literature contains several distinct responses to the difficulty of MVSK optimization. A first strand improves the statistical side of the problem through shrinkage, structured comoment estimation, or multifactor moment models \citep{MartelliniZiemann2010,BoudtLuPeeters2015,WangHuangLu2025}. A second strand changes the return model itself through skewed or heavy-tailed parametric families that make the optimization problem more tractable \citep{BirgeChavezBedoya2016,WangEtAl2023}. A related portfolio-construction strand studies tilting, convex scalarizations, or Pareto-front approximations that preserve higher-moment preferences while altering the optimization geometry \citep{BoudtCornillyVanHolleWillems2020,Steenkamp2023,LeCourtoisXu2024}. A third strand attacks the unrestricted quartic program directly, most notably through successive convex approximation and DC-based local methods \citep{PhamNiu2011,NiuWangThiPham2019,ZhouPalomar2021,ZhangNiu2024}. Recent surveys synthesize these routes and emphasize the persistent trade-off between statistical tractability, computational tractability, and fidelity to the unrestricted sample-moment model \citep{MandalThakur2024}. A fourth strand comes from affine geometry and affine-normal descent. \citet{ChengChengYau2005} introduced affine-normal search directions for minimization; \citet{NiuSheshmaniYau2026YAND} developed YAND as an algorithmic framework for smooth convex and nonconvex optimization with descent, convergence, and rate guarantees; and \citet{NiuSheshmaniYau2026LogDet} supplied a scalable matrix-free route for computing affine-normal directions through log-determinant tangent geometry and structured derivative actions.

Our paper sits closest to the third and fourth strands. Q-MVSK and the DC family are the natural unrestricted-MVSK baselines because they target the same investor-facing model. YAND provides a different computational route while preserving that model. Relative to the convex-scalarization and Pareto-front papers, the focus here is not efficient-set construction but direct large-scale solution of the scalarized unrestricted sample-moment quartic benchmark. Relative to the existing literature more broadly, the paper's distinguishing feature is the construction of a dedicated YAND solver for unrestricted sample-moment MVSK together with an MVSK-specific structural rate analysis that makes the reduced convergence constants explicit.

This positioning also clarifies the paper's scope. The paper does not aim to provide a broad benchmark across unrelated local directions or a general comparison of nonlinear-programming heuristics. Its focus is a YAND-based solver for unrestricted MVSK and the financial analysis that becomes possible once that solver is available at larger scales.

\section{A Structure-Exploiting YAND Solver for MVSK}
\label{sec:yand-algorithm}

\subsection{Exact sample-oracle formulation}

The key implementation choice is to keep unrestricted MVSK exact while refusing to build the dense quartic model explicitly. Let
\[
A:=R-\bm{1}\mu^\top \in \R^{T\times n},
\qquad
z(x):=Ax.
\]
Then the centered sample moments are functions of the projected centered return vector $z(x)$ alone. This means that the value, gradient, Hessian--vector products, and the directional third-order actions required by matrix-free YAND can all be computed from repeated multiplications by $A$ and $A^\top$, together with elementwise operations on $z(x)$ and projected directions. Subsection~\ref{subsec:algorithmic-summary} records the exact formulas.

This exact sample-oracle view is what makes YAND-MVSK computationally different from the explicit quartic route. The solver stores $(\mu,A,c)$ rather than covariance, coskewness, and cokurtosis tensors. At a fixed iterate it reuses the cached projection $z=Ax$ across all derivative queries generated inside the YAND direction routine. The dominant stored object is therefore the return matrix itself, not an explicit fourth-order polynomial representation.

\subsection{Affine-normal direction on the simplex}

YAND is naturally expressed on an unconstrained domain, so we solve the simplex-constrained problem through reduced coordinates. Fix an interior reference portfolio $x^{\mathrm{ref}} \in \operatorname{ri}(\simplex)$ and an orthonormal basis $U \in \R^{n\times (n-1)}$ of the simplex tangent space
\[
{\mathcal T}:=\{v \in \R^n : \bm{1}^\top v = 0\}.
\]
Writing
\[
x=x^{\mathrm{ref}}+Uy,
\qquad
\phi(y):=f(x^{\mathrm{ref}}+Uy),
\]
turns the interior simplex problem into an unconstrained reduced problem.

At a nonstationary reduced iterate $y$, the YAND step is assembled from the reduced gradient, a normal-aligned tangent basis, the tangent Hessian block, and an affine-normal correction term computed from the exact oracle actions summarized in Subsection~\ref{subsec:algorithmic-summary}. The method therefore computes the affine-normal step associated with the current MVSK level set and then lifts that step back to the simplex through the fixed basis $U$.
More concretely, if
\[
\bar g(y):=\nabla\phi(y),\qquad
\nu(y):=\frac{\bar g(y)}{\|\bar g(y)\|_2},
\]
and $Q_y\in\R^{(n-1)\times(n-2)}$ is the Householder frame satisfying $Q_y^\top Q_y=I_{n-2}$ and $Q_y^\top \nu(y)=0$, then the exact-trace YAND routine uses
\[
h(y):=Q_y^\top \nabla^2\phi(y)\nu(y),
\qquad
\mathcal H_{T,\lambda}(y)\eta:=Q_y^\top \nabla^2\phi(y)(Q_y\eta)+\lambda\eta,
\]
together with the exact log-determinant correction vector $a(y)$ returned by the reduced third-order oracle, and solves
\[
u(y)=\mathcal H_{T,\lambda}(y)^{-1}
\Bigl(h(y)-\frac{\|\bar g(y)\|_2}{n}a(y)\Bigr),
\qquad
d_y(y):=Q_yu(y)-\nu(y).
\]
The ambient simplex direction is then $d(x):=Ud_y(y)$. Thus the core YAND computation is a reduced linear solve driven by exact Hessian and third-order actions, not a black-box nonlinear-programming subroutine.

This geometric specialization matters because MVSK landscapes are typically anisotropic. The four preference coefficients can induce very different local scalings across mean, variance, skewness, and kurtosis. YAND is attractive precisely because its step is built from level-set affine geometry rather than from a purely Euclidean first-order model.

\subsection{Quartic line search and simplex boundary handling}

Along any fixed feasible tangent direction $d$, the unrestricted sample-moment MVSK objective satisfies
\[
\alpha \mapsto f(x+\alpha d),
\]
which is a univariate quartic polynomial on the feasible interval. This yields an exact step-selection rule: the line-search problem reduces to evaluating a quartic on a bounded interval and checking the real roots of its cubic derivative. Subsection~\ref{subsec:algorithmic-summary} formalizes this exact quartic line search from sample power sums.

Feasibility on the simplex is handled through an interior slice or active face. For $\tau \in [0,1/n)$, define
\[
\simplex(\tau):=\{x \in \simplex : x_i \ge \tau \text{ for all } i\},
\qquad
\alpha_{\max}(x,d;\tau):=
\min_{i:d_i<0}\frac{x_i-\tau}{-d_i}.
\]
The raw YAND step is first restricted to $[0,\alpha_{\max}(x,d;\tau)]$. If a nominal step still leaves the current simplex slice or exposed face, the implementation projects the trial point back to the feasible set, line-searches along the induced feasible segment, and activates a lower-dimensional face solver only when the projected segment still fails to provide a usable descent step. This safeguards the iteration against false termination at feasible but nonstationary boundary points.

\subsection{Algorithmic summary}
\label{subsec:algorithmic-summary}

Algorithm~\ref{alg:yand-mvsk} summarizes the resulting solver. To make clear what is actually computed at each iterate, we record the exact sample-oracle kernels that feed the YAND direction and the line search. At $x^k$, with $z^k:=Ax^k$, the exact first-, second-, and third-order actions are
\begin{align*}
g^k &= -c_1\mu + \frac{2c_2}{T}A^\top z^k - \frac{3c_3}{T}A^\top (z^k)^{\circ 2} + \frac{4c_4}{T}A^\top (z^k)^{\circ 3},\\
H_k v &= \frac{2c_2}{T}A^\top(Av) - \frac{6c_3}{T}A^\top(z^k \circ Av) + \frac{12c_4}{T}A^\top\bigl((z^k)^{\circ 2} \circ Av\bigr),\\
\mathcal T_k(u,v) &= -\frac{6c_3}{T}A^\top\!\bigl((Au)\circ(Av)\bigr) + \frac{24c_4}{T}A^\top\!\bigl(z^k \circ (Au)\circ(Av)\bigr).
\end{align*}
In reduced coordinates, with $y^k:=U^\top(x^k-x^{\mathrm{ref}})$, $\bar g^k:=U^\top g^k$, and $\nu^k:=\bar g^k/\|\bar g^k\|_2$, let $Q_k^\top Q_k=I_{n-2}$ and $Q_k^\top \nu^k=0$. The exact reduced YAND step is obtained from
\[
h^k:=Q_k^\top U^\top H_k(U\nu^k),\qquad
\mathcal H_{T,\lambda}^k\eta:=Q_k^\top U^\top H_k(UQ_k\eta)+\lambda\eta,
\]
\[
u^k=(\mathcal H_{T,\lambda}^k)^{-1}\Bigl(h^k-\frac{\|\bar g^k\|_2}{n}a^k\Bigr),
\qquad
d_y^k:=Q_ku^k-\nu^k,
\qquad
d^k:=Ud_y^k,
\]
where $a^k$ is the exact log-determinant correction vector computed from the reduced third-order oracle. Along $d^k$, with $w^k:=Ad^k$ and
\[
s_{rs}^k:=\frac1T\sum_{t=1}^T (z_t^k)^r (w_t^k)^s,
\qquad r+s\le 4,
\]
the line-restricted objective is
\[
\varphi_k(\alpha)=f(x^k+\alpha d^k)=A_0^k+A_1^k\alpha+A_2^k\alpha^2+A_3^k\alpha^3+A_4^k\alpha^4
\]
with
\[
A_0^k=f(x^k),\qquad
A_1^k=-c_1\mu^\top d^k+2c_2s_{11}^k-3c_3s_{21}^k+4c_4s_{31}^k,
\]
\[
A_2^k=c_2s_{02}^k-3c_3s_{12}^k+6c_4s_{22}^k,\qquad
A_3^k=-c_3s_{03}^k+4c_4s_{13}^k,\qquad
A_4^k=c_4s_{04}^k.
\]
In these formulas, $\nu^k$ is the normalized reduced gradient, $Q_k$ removes the normal direction and keeps only tangent motion along the current level set, $\mathcal H_{T,\lambda}^k$ is the regularized reduced tangent Hessian operator, and $a^k$ is the affine-normal correction. Thus the YAND step is assembled by solving for the tangential component $u^k$, reinserting the normal component through $d_y^k=Q_ku^k-\nu^k$, and then minimizing the resulting quartic on $[0,\alpha_{\max}(x^k,d^k;\tau)]$. The iteration therefore has four blocks: exact sample-oracle evaluation, affine-normal direction construction, exact quartic step selection, and boundary-aware acceptance.

\begin{algorithm}[ht]
\caption{YAND-MVSK algorithm.}
\label{alg:yand-mvsk}
\begin{algorithmic}[1]
\State \textbf{Input:} return matrix $R$, preference vector $c$, interior start $x^0 \in \operatorname{ri}(\simplex)$, interior margin $\tau$, tolerance $\varepsilon$
\State Compute $\mu$, $A=R-\bm{1}\mu^\top$, set $x^{\mathrm{ref}}:=x^0$, and choose an orthonormal tangent basis $U$
\For{$k=0,1,2,\ldots$}
    \Statex \textbf{Oracle block}
    \State Evaluate the exact sample oracle at $x^k$: form $z^k=Ax^k$, $g^k=\nabla f(x^k)$, and $\bar g^k=U^\top g^k$
    \If{$\|\bar g^k\|_2\le \varepsilon$}
        \State \textbf{break}
    \EndIf

    \Statex \textbf{Affine-normal block}
    \State Normalize the reduced gradient, build the Householder frame $Q_k$, assemble the reduced tangent system, and solve for the YAND direction $d^k$

    \Statex \textbf{Quartic step block}
    \State Form the quartic line model $\varphi_k(\alpha)=f(x^k+\alpha d^k)$ on $[0,\alpha_{\max}(x^k,d^k;\tau)]$
    \State Choose $\alpha_k\in\arg\min \varphi_k(\alpha)$ by checking the interval endpoints and the real roots of $\varphi_k'$, and set $\tilde x^{k+1}=x^k+\alpha_k d^k$

    \Statex \textbf{Boundary block}
    \If{$\tilde x^{k+1}$ leaves the current simplex slice or face}
        \State Project the trial back to the feasible set and line-search on the induced feasible segment
        \If{the corrected segment still fails to yield a positive descent step}
            \State Continue the iteration on the newly exposed active face
        \EndIf
    \EndIf
    \State Accept the resulting feasible point as $x^{k+1}$
\EndFor
\State \textbf{Output:} $x^\star:=x^k$
\end{algorithmic}
\end{algorithm}

\section{Structural Analysis and Convergence Guarantees}
\label{sec:theory}

This section formalizes the exact sample-oracle YAND view of unrestricted MVSK. The appendix records the omitted operator proofs, reduced-coordinate derivations, and wrapper-level implementation details.

Let $R \in \R^{T \times n}$ denote the sample-return matrix, let
\[
A:=R-\bm{1}\mu^\top \in \R^{T \times n}
\]
be the centered sample-return matrix, and define
\[
z(x):=Ax \in \R^T.
\]
For a vector $u \in \R^T$, we write $u_t$ for its $t$th component and $\hadpow{u}{p}$ for its componentwise $p$th power. We use $\|\cdot\|_2$ for the Euclidean norm and $\|\cdot\|_{\mathrm{op}}$ for the spectral operator norm.

\subsection{Operator Structure and Curvature Decomposition}

\begin{proposition}[Exact operator representation]
\label{prop:oracle-representation}
For every $x \in \simplex$, the unrestricted sample-moment MVSK objective can be written exactly as
\begin{equation}
\label{eq:operator-f}
f(x)= -c_1 \mu^\top x + \frac{c_2}{T}\|z(x)\|_2^2 - \frac{c_3}{T}\sum_{t=1}^T z_t(x)^3 + \frac{c_4}{T}\sum_{t=1}^T z_t(x)^4.
\end{equation}
Moreover,
\begin{equation}
\label{eq:operator-grad}
\nabla f(x)= -c_1\mu + \frac{2c_2}{T}A^\top z - \frac{3c_3}{T}A^\top \hadpow{z}{2} + \frac{4c_4}{T}A^\top \hadpow{z}{3},
\end{equation}
and for every direction $v \in \R^n$,
\begin{equation}
\label{eq:operator-hvp}
\nabla^2 f(x)v = \frac{2c_2}{T}A^\top(Av) - \frac{6c_3}{T}A^\top(z \circ Av) + \frac{12c_4}{T}A^\top(\hadpow{z}{2} \circ Av),
\end{equation}
where $z=z(x)$. Finally, if for directions $u,v \in \R^n$ we write
\[
\mathcal T_3(x;u,v):=D^3f(x)[u,v,\cdot] \in \R^n,
\]
then
\begin{equation}
\label{eq:operator-third}
\mathcal T_3(x;u,v)= -\frac{6c_3}{T}A^\top\!\bigl((Au)\circ(Av)\bigr) + \frac{24c_4}{T}A^\top\!\bigl(z \circ (Au)\circ(Av)\bigr).
\end{equation}
\end{proposition}

Proposition~\ref{prop:oracle-representation} is the computational core of YAND-MVSK. Once the objective is written through projected centered returns, every oracle quantity required by the matrix-free YAND routine is obtained from matrix--vector operations with $A$ and $A^\top$ together with elementwise vector algebra.

\begin{proposition}[Lifted separability, curvature decomposition, and tangential conditioning]
\label{prop:curvature-decomposition}
Define
\[
\psi(s):=c_2 s^2-c_3 s^3+c_4 s^4,
\qquad
\Phi(z):=\frac{1}{T}\sum_{t=1}^T \psi(z_t).
\]
Then
\[
f(x)=-c_1\mu^\top x+\Phi(Ax).
\]
Let
\[
{\mathcal T}:=\{v\in\R^n:\bm{1}^\top v=0\}
\]
denote the simplex tangent space.
Moreover, for every $x\in\simplex$,
\[
\nabla^2 f(x)=\frac{1}{T}A^\top D(x)A,
\qquad
D(x):=\operatorname{Diag}\!\bigl(\psi''(z_1(x)),\ldots,\psi''(z_T(x))\bigr),
\]
where
\[
\psi''(s)=2c_2-6c_3 s+12c_4 s^2.
\]
If $\mathcal X\subset\simplex$ is such that, for some constant $\widehat\gamma\ge 0$,
\[
\bigl|\psi''\bigl((Ax)_t\bigr)\bigr|\le \widehat\gamma
\qquad
\text{for all }x\in\mathcal X,\; t=1,\ldots,T,
\]
then for every $x\in\mathcal X$ and every tangent direction $d\in{\mathcal T}$,
\[
\bigl|d^\top \nabla^2 f(x)d\bigr|
\le
\frac{\widehat\gamma}{T}\|Ad\|_2^2.
\]
If, in addition, for some constants $0<\underline\gamma\le \overline\gamma$,
\[
\underline\gamma\le \psi''\bigl((Ax)_t\bigr)\le \overline\gamma
\qquad
\text{for all }x\in\mathcal X,\; t=1,\ldots,T,
\]
then for every $x\in\mathcal X$ and every tangent direction $d\in{\mathcal T}$,
\[
\frac{\underline\gamma}{T}\|Ad\|_2^2
\le
d^\top \nabla^2 f(x)d
\le
\frac{\overline\gamma}{T}\|Ad\|_2^2.
\]
Consequently, whenever $A$ is injective on ${\mathcal T}$,
\[
\kappa_{{\mathcal T}}\!\bigl(\nabla^2 f(x)\bigr)
\le
\frac{\overline\gamma}{\underline\gamma}\,\kappa_{{\mathcal T}}(A)^2,
\]
where
\[
\kappa_{{\mathcal T}}(A):=
\frac{\max\{\|Ad\|_2:\ d\in{\mathcal T},\ \|d\|_2=1\}}
{\min\{\|Ad\|_2:\ d\in{\mathcal T},\ \|d\|_2=1\}}.
\]
\end{proposition}

\begin{remark}[What this decomposition explains]
\label{rem:curvature-decomposition}
Proposition~\ref{prop:curvature-decomposition} has two layers. First, the exact factorization $\nabla^2 f(x)=T^{-1}A^\top D(x)A$ and the absolute curvature bound show that tangential curvature magnitudes are filtered through the data map $A$ and weighted by the scalar responses $\psi''((Ax)_t)$. This decomposition remains valid on nonconvex regions: the matrix $A$ controls the Euclidean anisotropy induced by the data, while $\psi$ controls the intrinsic nonconvexity, sign changes, and local curvature variation created by skewness and kurtosis.

Second, on regions where $\psi''$ is uniformly positive, the same decomposition sharpens to a genuine tangential condition-number bound, separating the curvature spread $\overline\gamma/\underline\gamma$ from the geometric distortion $\kappa_{\mathcal T}(A)^2$. Outside such sign-definite regions, even if $A$ is injective on ${\mathcal T}$, mixed signs in $D(x)$ can make the tangential Hessian indefinite or singular, so no finite positive-definite condition number need exist.

This does not mean that YAND removes the nonconvexity of MVSK. It does explain why affine-normal descent is a natural computational match: the method is designed to be insensitive to affine distortion of level sets, so the $A$-induced part of the ill-conditioning is precisely the part that should matter less under the affine-normal geometry than under a purely Euclidean method. This interpretation is consistent with the affine-scaling theory in \citet[Section~9]{NiuSheshmaniYau2026YAND}: when ill-conditioning is created purely by anisotropic affine scaling, the mapped YAND dynamics are invariant and the relevant convergence constants are inherited from the unscaled base objective. Proposition~\ref{prop:curvature-decomposition} does not claim that unrestricted MVSK reduces completely to that setting, but it does identify the same structural mechanism on which those YAND robustness results act.

Because $\bm{1}_T^\top A=0$, one has $\operatorname{rank}(A)\le T-1$, so injectivity on ${\mathcal T}$ can occur on the full simplex only if $n\le T$. If $A$ is not injective on ${\mathcal T}$, then there exists a nonzero tangent direction $d\in{\mathcal T}$ with $Ad=0$, and therefore $\nabla^2 f(x)d=0$ for every $x$. Thus the tangential Hessian is singular on the full simplex, and the condition-number bound above should be read as meaningful only on tangent subspaces, or reduced faces, where the relevant data map is injective.
\end{remark}

\begin{corollary}[Matrix-free complexity]
\label{cor:oracle-complexity}
Suppose the centered sample matrix $A$ is stored explicitly. Then evaluating $f(x)$, $\nabla f(x)$, one Hessian--vector product $\nabla^2 f(x)v$, and one directional third-order action $\mathcal T_3(x;u,v)$ from \eqref{eq:operator-f}--\eqref{eq:operator-third} each requires $O(Tn)$ arithmetic operations and $O(T+n)$ working memory. By contrast, storing the full coskewness and cokurtosis tensors requires $\Theta(n^3+n^4)$ coefficients before optimization begins.
\end{corollary}

\subsection{Feasible Quartic Line Search and Simplex Reduction}

\begin{proposition}[Quartic exact line search from sample power sums]
\label{prop:quartic-exact-linesearch}
Fix $x\in \simplex(\tau)$ and a nonzero tangent direction $d\in {\mathcal T}$. Let
\[
z:=Ax,
\qquad
w:=Ad,
\]
and define the mixed sample power sums
\[
s_{rs}(x,d):=\frac1T\sum_{t=1}^T z_t^r w_t^s,
\qquad
r,s\in \mathbb{N}\cup\{0\},
\quad
r+s\le 4.
\]
Then the line-restricted objective on the feasible interval $[0,\alpha_{\max}(x,d;\tau)]$ can be written as
\[
\varphi_{x,d}(\alpha):=f(x+\alpha d)=A_0+A_1\alpha+A_2\alpha^2+A_3\alpha^3+A_4\alpha^4,
\]
where
\[
A_0=f(x),
\qquad
A_1=-c_1\mu^\top d+2c_2 s_{11}(x,d)-3c_3 s_{21}(x,d)+4c_4 s_{31}(x,d),
\]
\[
A_2=c_2 s_{02}(x,d)-3c_3 s_{12}(x,d)+6c_4 s_{22}(x,d),
\]
\[
A_3=-c_3 s_{03}(x,d)+4c_4 s_{13}(x,d),
\qquad
A_4=c_4 s_{04}(x,d).
\]
Consequently, $\varphi_{x,d}'(\alpha)$ is a cubic polynomial, and an exact feasible line-search step is obtained by checking the interval endpoints together with the real roots of $\varphi_{x,d}'$ that lie in $(0,\alpha_{\max}(x,d;\tau))$.
\end{proposition}

\begin{remark}[Practical interpretation]
\label{rem:quartic-exact-linesearch}
Proposition~\ref{prop:quartic-exact-linesearch} is a structural statement. In implementation one does not need a symbolic quartic package; one needs only the cached projection $z=Ax$, one additional projected direction $w=Ad$, the scalar sums $s_{rs}(x,d)$, and a stable cubic-root routine. Accordingly, exact line search is a natural default in YAND-MVSK rather than an ancillary refinement.
\end{remark}

\begin{definition}[Tangent-gradient residual on the simplex interior]
\label{def:stationarity}
Let
\[
{\mathcal T}:=\{v \in \R^n : \bm{1}^{\top}v=0\},
\qquad
P_{\mathcal T}:=I_n-\frac{1}{n}\bm{1}\bm{1}^{\top}.
\]
For $x \in \operatorname{ri}(\simplex)$, define the tangent-gradient residual
\[
G_{\mathcal T}(x):=P_{\mathcal T}\nabla f(x).
\]
\end{definition}

\begin{proposition}[Interior stationarity criterion]
\label{prop:interior-stationarity}
For $x \in \operatorname{ri}(\simplex)$,
\[
G_{\mathcal T}(x)=0
\qquad\Longleftrightarrow\qquad
x \text{ is a first-order stationary point of \eqref{eq:mvsk}.}
\]
\end{proposition}

\begin{proposition}[Affine-subspace reduction]
\label{prop:simplex-reduction}
Fix $x^{\mathrm{ref}} \in \operatorname{ri}(\simplex)$ and let $U \in \R^{n \times (n-1)}$ have orthonormal columns spanning ${\mathcal T}$. Define
\[
\mathcal Y:=\{y \in \R^{n-1} : x^{\mathrm{ref}}+Uy \in \operatorname{ri}(\simplex)\},
\qquad
\phi(y):=f(x^{\mathrm{ref}}+Uy).
\]
Then, for every $x=x^{\mathrm{ref}}+Uy \in \operatorname{ri}(\simplex)$,
\[
\nabla \phi(y)=U^\top \nabla f(x),
\qquad
\nabla^2 \phi(y)=U^\top \nabla^2 f(x)U.
\]
Moreover,
\[
G_{\mathcal T}(x)=0
\qquad\Longleftrightarrow\qquad
\nabla \phi(y)=0.
\]
\end{proposition}

\begin{proposition}[Reduced-coordinate exact oracle]
\label{prop:reduced-oracle}
Under Proposition~\ref{prop:simplex-reduction}, define for $y \in \mathcal Y$ and directions $u,v \in \R^{n-1}$
\[
\mathcal T_{3,\phi}(y;u,v):=U^\top \mathcal T_3(x;Uu,Uv),
\qquad
x=x^{\mathrm{ref}}+Uy.
\]
Then the reduced objective $\phi$ admits exact oracle kernels
\[
\nabla \phi(y)=U^\top \nabla f(x),
\qquad
\nabla^2 \phi(y)u=U^\top \nabla^2 f(x)(Uu),
\]
\[
\mathcal T_{3,\phi}(y;u,v)=U^\top \mathcal T_3(x;Uu,Uv).
\]
If $A$ is stored explicitly and the tangent basis $U$ is represented densely, one evaluation of $\phi(y)$, $\nabla \phi(y)$, one reduced Hessian--vector action, and one reduced directional third-order action each requires $O(Tn+n^2)$ arithmetic operations.
\end{proposition}

\subsection{Regularity, Convexity Certificates, YAND Transfer, and Structural PL}

\begin{proposition}[Explicit MVSK regularity constants on an interior simplex slice]
\label{prop:mvsk-explicit-constants}
Fix $\tau \in (0,1/n)$ and define
\[
B_\tau:=\sup_{x\in\simplex(\tau)}\|Ax\|_\infty.
\]
Then for every $x,y\in \simplex(\tau)$,
\[
\|\nabla^2 f(x)\|_{\mathrm{op}}
\le
L_\tau:=
\frac{\|A\|_{\mathrm{op}}^2}{T}
\bigl(2c_2+6c_3B_\tau+12c_4B_\tau^2\bigr),
\]
and
\[
\|\nabla^2 f(x)-\nabla^2 f(y)\|_{\mathrm{op}}
\le
M_\tau\|x-y\|_2,
\qquad
M_\tau:=
\frac{\|A\|_{\mathrm{op}}^3}{T}
\bigl(6c_3+24c_4B_\tau\bigr).
\]
Thus $f$ is $L_\tau$-smooth and has $M_\tau$-Lipschitz Hessian on $\simplex(\tau)$. Moreover, because $\|x\|_2\le 1$ on $\simplex$, one always has
\[
B_\tau\le \|A\|_{\mathrm{op}},
\]
so these constants are controlled explicitly by $\|A\|_{\mathrm{op}}$ and $c$ alone.
\end{proposition}

Two additional technical closure results are deferred to the appendix. First, once the iterates are safeguarded on an interior slice $\simplex(\tau)$, the smoothness and bounded-below requirements used below become automatic for the polynomial MVSK objective. Second, the appendix records a sufficient condition under which the imported reduced-coordinate bounded-direction hypothesis from \citet{NiuSheshmaniYau2026YAND} is satisfied. We keep these slice-closure and bounded-direction details in the appendix so that the main text can focus on the transfer statement itself and on the explicit MVSK constants that enter it.

\begin{assumption}[Safeguarded interior YAND regularity]
\label{ass:yand-transfer}
Consider an exact-oracle YAND sequence $\{x^k\}\subset \operatorname{ri}(\simplex)$ generated by feasible updates
\[
x^{k+1}=x^k+\alpha_k U d_y^k,
\]
where $U$ is as in Proposition~\ref{prop:simplex-reduction}, $d_y^k$ is the reduced-coordinate YAND direction for $\phi$, and $\alpha_k$ is chosen by Armijo backtracking on $\phi$. Assume that:
\begin{enumerate}
\item there exists a compact set $K\subset \operatorname{ri}(\simplex)$ containing all iterates, and $f$ is $L$-smooth and bounded below on a neighborhood of $K$;
\item the reduced-coordinate YAND directions satisfy the bounded-direction condition of \citet[Assumption~7.1]{NiuSheshmaniYau2026YAND} with some constant $\beta_{\mathrm{dir}}<\infty$; in particular, their angle with $-\nabla \phi(y^k)$ is uniformly bounded away from $\pi/2$;
\item the accepted Armijo steps, computed with parameter $\sigma\in(0,1)$, remain in $K$.
\end{enumerate}
\end{assumption}

\begin{theorem}[YAND transfer on the simplex interior]
\label{thm:yand-stationarity}
Under Assumption~\ref{ass:yand-transfer}, either the exact-oracle YAND method terminates at an interior first-order stationary point of \eqref{eq:mvsk}, or
\[
\lim_{k\to\infty}\|G_{\mathcal T}(x^k)\|_2 = 0.
\]
If $\{x^k\}$ has an accumulation point $\bar x \in K$, then $\bar x$ is a first-order stationary point of \eqref{eq:mvsk}. If, in addition, the reduced objective $\phi$ satisfies the Polyak--Lojasiewicz inequality on
\[
\mathcal Y_K:=\{y : x^{\mathrm{ref}}+Uy \in K\},
\]
namely,
\[
\frac12\|\nabla \phi(y)\|_2^2
\ge
\mu_{\mathrm{PL}}\bigl(\phi(y)-\phi^\star\bigr)
\qquad
\text{for all } y\in\mathcal Y_K,
\]
with constant $\mu_{\mathrm{PL}}>0$ and $\phi^\star:=\inf_{y\in\mathcal Y_K}\phi(y)$, then
\[
f(x^{k+1})-f^\star
\le
\bigl(1-\rho_{\mathrm{PL}}\bigr)\bigl(f(x^k)-f^\star\bigr),
\qquad
\rho_{\mathrm{PL}}=
\frac{2\sigma(1-\sigma)\mu_{\mathrm{PL}}}{L(1+\beta_{\mathrm{dir}}^2)},
\]
where $f^\star:=\inf_{x\in K}f(x)$. Under the local nondegeneracy and step-size assumptions of \citet[Assumption~7.14, Theorem~7.16, and Remark~7.17]{NiuSheshmaniYau2026YAND}, the same reduced-coordinate transfer also yields local quadratic convergence near an interior nondegenerate minimizer.
\end{theorem}

\begin{corollary}[A low-dimensional structural PL regime for reduced MVSK]
\label{cor:mvsk-structural-pl}
Fix $\tau \in (0,1/n)$ and $x^{\mathrm{ref}}\in\simplex(\tau)$, let $U\in\R^{n\times(n-1)}$ have orthonormal columns spanning ${\mathcal T}$, and define
\[
\mathcal Y_\tau:=\{y\in\R^{n-1}:x^{\mathrm{ref}}+Uy\in\simplex(\tau)\},
\qquad
\phi(y):=f(x^{\mathrm{ref}}+Uy).
\]
Let
\[
\underline\gamma_\tau:=\min_{|s|\le B_\tau}\psi''(s),
\qquad
\overline\gamma_\tau:=\max_{|s|\le B_\tau}\psi''(s).
\]
If $\underline\gamma_\tau>0$, $n\le T$, and $A$ is injective on ${\mathcal T}$, then for every $y\in\mathcal Y_\tau$,
\[
\mu_\tau I_{n-1}
\preceq
\nabla^2\phi(y)
\preceq
L_{\tau,\phi} I_{n-1},
\]
where
\[
\mu_\tau:=\frac{\underline\gamma_\tau}{T}\sigma_{\min}(AU)^2,
\qquad
L_{\tau,\phi}:=\frac{\overline\gamma_\tau}{T}\sigma_{\max}(AU)^2.
\]
Hence $\phi$ is $\mu_\tau$-strongly convex on $\mathcal Y_\tau$ and satisfies the Polyak--Lojasiewicz inequality there with the same constant. Under Assumption~\ref{ass:yand-transfer}, Theorem~\ref{thm:yand-stationarity} therefore yields the explicit linear decrease estimate
\[
f(x^{k+1})-f^\star
\le
\bigl(1-\rho_\tau\bigr)\bigl(f(x^k)-f^\star\bigr),
\qquad
\rho_\tau=
\frac{2\sigma(1-\sigma)\mu_\tau}{L_{\tau,\phi}(1+\beta_{\mathrm{dir}}^2)}.
\]
\end{corollary}

\begin{corollary}[A coefficient-only convexity certificate for MVSK]
\label{cor:mvsk-convexity}
A sufficient coefficient-only condition for $\underline\gamma_\tau>0$ on every slice is
\[
c_4>0,
\qquad
8c_2c_4>3c_3^2,
\]
in which case $\psi''(s)\ge 2c_2-\frac{3c_3^2}{4c_4}>0$ for all $s\in\R$. Hence $\psi$ is globally strictly convex, $\Phi(z)=T^{-1}\sum_{t=1}^T\psi(z_t)$ is convex on $\R^T$, and therefore
\[
f(x)=-c_1\mu^\top x+\Phi(Ax)
\]
is a convex objective on $\simplex$. Thus \eqref{eq:mvsk} admits a coefficient-only sufficient test for convexity that can be checked before looking at the data matrix. In particular, the same coefficient condition implies $\underline\gamma_\tau>0$ on every interior slice $\simplex(\tau)$.
\end{corollary}

\begin{remark}[Standard CRRA calibration lies in the convex regime]
\label{rem:mvsk-crra}
Up to a positive common scaling, a standard utility-based CRRA calibration in the MVSK literature takes the form
\[
(c_1,c_2,c_3,c_4)=
\left(
1,\;
\frac{\gamma}{2},\;
\frac{\gamma(\gamma+1)}{6},\;
\frac{\gamma(\gamma+1)(\gamma+2)}{24}
\right),
\qquad
\gamma>0,
\]
where $\gamma$ is the relative risk-aversion parameter \citep{Palomar2025}. Then
\[
c_4=\frac{\gamma(\gamma+1)(\gamma+2)}{24}>0,
\qquad
8c_2c_4-3c_3^2
=
\frac{\gamma^2(\gamma+1)(\gamma+3)}{12}>0,
\]
so Corollary~\ref{cor:mvsk-convexity} applies automatically. Hence the investor-facing MVSK objective is convex on $\simplex$ under the standard CRRA-induced coefficients. In that regime, any failure of slice-wise strong convexity or of the structural PL condition must come from the data geometry, not from the preference specification.
\end{remark}

\begin{remark}[Interpretation]
\label{rem:yand-transfer}
Theorem~\ref{thm:yand-stationarity} still imports the general YAND convergence framework, but Proposition~\ref{prop:mvsk-explicit-constants} together with Corollaries~\ref{cor:mvsk-structural-pl} and \ref{cor:mvsk-convexity} turn its abstract constants into MVSK quantities. The smoothness and curvature constants factor into a data-geometry term, through $\|A\|_{\mathrm{op}}$ and the singular values of $AU$, and a preference-geometry term, through the one-dimensional quartic response $\psi$. This separation distinguishes nonconvexity created by the investor coefficients from degeneracy created by the data map. Corollary~\ref{cor:mvsk-convexity} identifies a regime in which nonconvexity is ruled out by the preference coefficients alone, before any rank condition on $A$ is imposed, and Remark~\ref{rem:mvsk-crra} shows that the standard CRRA calibration lies in this regime automatically. Corollary~\ref{cor:mvsk-structural-pl} then shows that, within this convex regime, slice-wise strong convexity and the structural PL inequality require an additional data-geometry condition, namely injectivity of $A$ on the relevant tangent space. Because $\operatorname{rank}(A)\le T-1$, the resulting PL conclusion is inherently low-dimensional on the full simplex and should be read as a face-restricted possibility in the high-dimensional regimes studied computationally. This is the paper's main theoretical sharpening beyond a pure transfer argument.
\end{remark}

\section{Computational Study}
\label{sec:computational-evidence}

\subsection{Experimental design}

The computational study answers three questions. First, what happens when data geometry deteriorates inside the coefficient-certified convex regime of Corollary~\ref{cor:mvsk-convexity}? Second, how should YAND-MVSK be configured as the asset dimension grows on a standard unrestricted-MVSK benchmark? Third, once unrestricted MVSK becomes computationally feasible at full-universe scale, does it materially change the comparison with exact MV in a realistic asset universe? This separation is deliberate. Once the investor coefficients themselves guarantee convexity, the remaining difficulty comes from the singular geometry of $AU$, not from competing local minima.

The first part is a controlled CRRA conditioning benchmark. We fix a representative CRRA calibration and worsen the positive-spectrum conditioning of the reduced data map $AU$ while keeping $T>n$. Its role is diagnostic rather than financial: once CRRA places MVSK in the coefficient-certified convex regime, deterioration in runtime or projected-simplex stationarity can be attributed to data geometry rather than to competing local minima. Because plug-in higher moments are easier to interpret in a sample-rich regime, the main paper emphasizes this $T>n$ design. A complementary $n>T$ counterpart is retained in the Appendix as a computational stress test.

The second part is the standard synthetic benchmark at $T=252$. It is used both for cross-solver comparison on the largest common grid and for identifying when the YAND implementation should switch from the direct reduced solve to the PCG-based large-scale regime.

For each asset dimension $n$ and sample length $T$, we draw a return matrix $R\in\R^{T\times n}$ with independent entries on $[-0.1,0.4]$, compute the sample mean $\mu=T^{-1}R^\top\bm{1}$, form the centered matrix $A:=R-\bm{1}\mu^\top$, and start every run from the equal-weight portfolio $x^0=\bm{1}/n$. Each return matrix is paired with three stylized coefficient-stress profiles,
\[
[10,1,10,1], \qquad [1,10,1,10], \qquad [10,10,10,10],
\]
representing return-seeking, risk-averse, and balanced scalarizations. Because a common positive scaling of $c$ does not change the optimizer, these are better interpreted as stress profiles than as economically calibrated investor types. Under Corollary~\ref{cor:mvsk-convexity}, the latter two lie in the coefficient-certified convex regime, whereas the return-seeking profile does not. The second part therefore spans both coefficient-certified convex and potentially nonconvex unrestricted-MVSK instances.

The synthetic benchmark spans asset dimensions from $n=4$ to $n=5000$, covering both the common comparison grid and the dimensional scale reached in the real-data application. This second part is used for solver validation and regime diagnosis rather than for the paper's financial claim.

The third part is the full-scale real-data study of Section~\ref{sec:when-higher-moments-matter}. Its role is economic rather than diagnostic: once unrestricted MVSK becomes computationally feasible at the full-universe scale, one can compare MVSK and exact MV directly under the same asset universe and return targets and ask when higher moments matter economically after diversification.

The first two layers are developed in this section, and Section~\ref{sec:when-higher-moments-matter} then turns to the economic layer. The Appendix complements them with the under-sampled stress counterpart, reduced-Hessian geometry diagnostics, parameter and line-search tables, the common-overlap large-scale comparison, and implementation checks, so that the solver crossover documented below can be read as a geometry effect rather than as a coding-constant artifact.

\subsection{Sample-Rich CRRA Conditioning Benchmark}

We instantiate the first empirical layer with a controlled convex benchmark at the representative CRRA calibration $\gamma=6$, which corresponds to
\[
(c_1,c_2,c_3,c_4)=\left(1,3,7,14\right).
\]
For each target conditioning level $\kappa_+(AU)\in\{1,10,10^2,10^3\}$, we construct a centered sample map whose positive singular spectrum on $AU$ attains the prescribed level while keeping $n=1000<T=2000$ and the equal-weight start in the null space of $A$. The resulting benchmark is sample-rich, coefficient-certified convex, and tangentially nondegenerate on the simplex tangent space. In that sense it isolates data geometry in the cleanest regime considered in the paper. The Appendix records the explicit matrix construction, the corresponding reduced-Hessian identity at the starting point, and the exact medians underlying this block.

Only YAND-MVSK (large) and Q-MVSK are compared in this conditioning benchmark, and each solver receives one untimed warm-up pass before three timed replications at every conditioning level. The PCG-based YAND configuration uses the same stall-recovery wrapper as in the large-scale synthetic study, with both mild and aggressive retries enabled so that a rare stalled run does not dominate the reported medians. This restriction reflects the regime evidence: the study below already shows that the direct YAND configuration is no longer the preferred implementation by the upper end of the hundreds, well before the thousand-asset block considered here, and the DC baselines are not the natural emphasis of a convex conditioning diagnostic whose purpose is to isolate data geometry once preference-driven nonconvexity has been removed.

\begin{figure*}[t]
\centering
\includegraphics[width=0.98\textwidth]{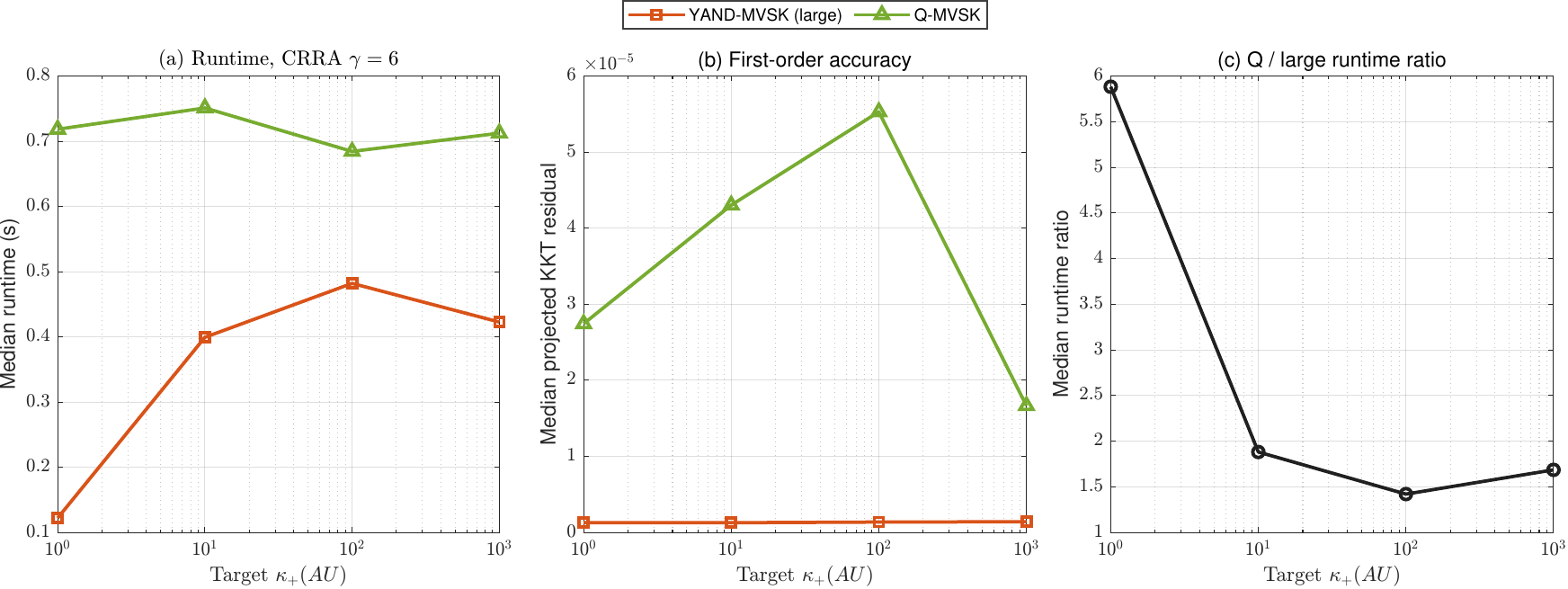}
\caption{Sample-rich controlled CRRA conditioning benchmark with representative calibration $\gamma=6$, $n=1000$, and $T=2000$. Each solver receives one untimed warm-up pass before three timed replications at every prescribed conditioning level. Panel (a) reports median runtime, panel (b) reports median projected-simplex KKT residual, and panel (c) reports the median runtime ratio Q-MVSK / YAND-MVSK (large).}
\label{fig:crra-conditioning}
\end{figure*}

Figure~\ref{fig:crra-conditioning} shows the conditioning story directly in the sample-rich regime that matters most for judging the investor-facing model. Across the full sweep $\kappa_+(AU)\in\{1,10,10^2,10^3\}$, the PCG-based YAND configuration keeps median runtime between 0.122 and 0.482 seconds and median projected-simplex residual between $1.27\times 10^{-6}$ and $1.40\times 10^{-6}$. By contrast, Q-MVSK requires between 0.684 and 0.751 seconds and remains between $1.67\times 10^{-5}$ and $5.53\times 10^{-5}$ in median projected-simplex residual. Panel (c) makes the runtime separation explicit: the median runtime ratio Q-MVSK / YAND-MVSK (large) stays between 1.4 and 5.9 over the conditioning sweep.

The benchmark therefore exhibits the desired pattern inside a coefficient-certified convex regime and in a sample-rich window: YAND-MVSK (large) is simultaneously faster and more accurate in first-order terms on an ill-conditioned unrestricted-MVSK block. Moreover, this is not a case in which Q-MVSK spends more time to buy tighter accuracy. Under the benchmark wrapper used here, Q-MVSK triggers its legacy step/objective stagnation safeguard after only 2 outer QP updates on every stored CRRA instance, before meeting the common projected-KKT target; the appendix sensitivity analysis shows that tightening those safeguards on representative CRRA and large-scale instances raises Q-MVSK runtime but leaves the projected-simplex residual well above the YAND values. Once preference-driven nonconvexity has been removed by the CRRA calibration, the remaining obstacle is purely data geometry, and this conditioning benchmark shows that YAND's advantage is already visible there before one moves to the harsher under-sampled regime. A complementary $n=5000>T=252$ CRRA stress test reported in the Appendix amplifies the same numerical separation, but that block should be read as a computational stress regime rather than as the paper's primary evidence on the investment meaning of plug-in higher moments.

\subsection{Synthetic Coefficient-Stress Benchmark Results}

\paragraph{Experimental protocol and solver settings.}
The reported results in the remainder of this section concern the second, stylized coefficient-stress part of the study. We use two YAND-MVSK configurations: a direct reduced solver on the $n\le 100$ benchmark grid and a PCG-based large-scale solver with stall recovery beyond that range. Both keep the exact sample oracle, quartic line search, projected active-set correction, and face continuation. Solution quality is evaluated ex post by the projected-simplex KKT residual, and the benchmark uses a common $10^{-6}$ target together with each baseline's native stall safeguards. The direct-versus-PCG regime comparison is reported through $n=1000$; beyond that point the synthetic study continues with YAND-MVSK (large) and the competing baselines. Full parameter settings, iteration budgets, restart rules, and the direct Armijo diagnostic are reported in the Appendix.

\begin{figure*}[t]
\centering
\includegraphics[width=0.98\textwidth]{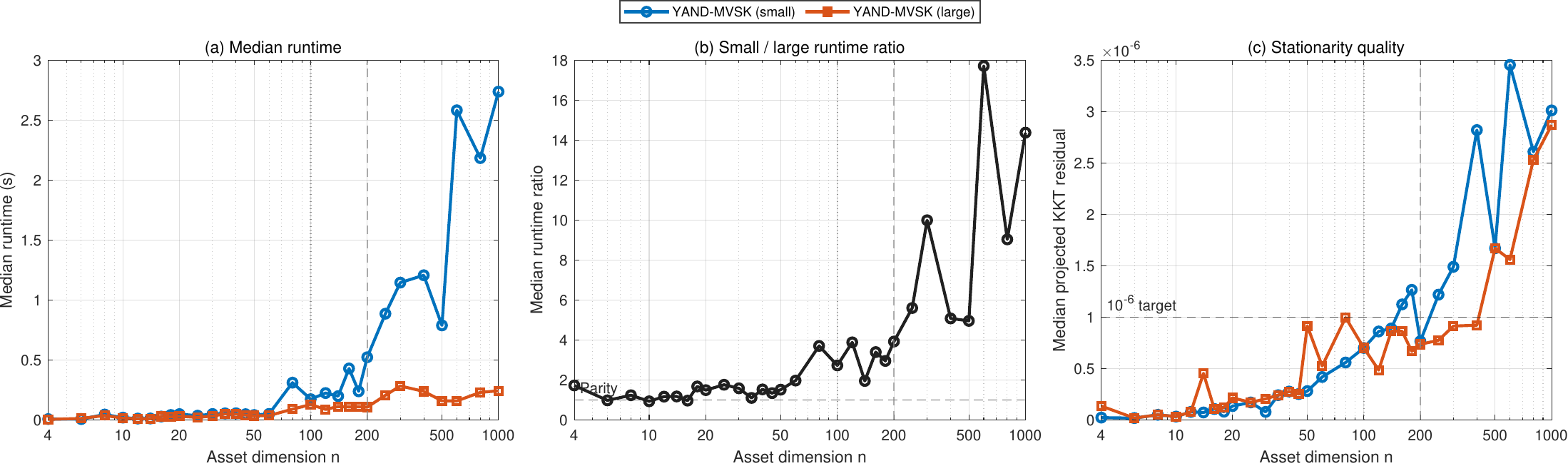}
\caption{YAND-MVSK performance by configuration on the synthetic benchmark at $T=252$. Panels report median runtime, the median instance-wise runtime ratio, and the median projected-simplex KKT residual. For $n>100$, the large-configuration curve uses stall recovery.}
\label{fig:yand-config-regimes}
\end{figure*}

\paragraph{YAND-MVSK regime comparison.}
Figure~\ref{fig:yand-config-regimes} and the supporting table in the Appendix summarize the regime split through $n=1000$. The appendix table reports the exact regime-by-regime medians behind the figure. Up to $n\le 60$, the two configurations remain close in objective value and projected-KKT quality, but no longer in runtime: the median times are $4.24\times 10^{-2}$ and $3.04\times 10^{-2}$ seconds, projected-KKT hits remain 48/48 versus 46/48, and objective outcomes are almost entirely ties. The range $80\le n\le 100$ is still the main crossover zone: YAND-MVSK (large) is already 2.74$\times$ faster, but YAND-MVSK (small) still records 6/6 projected-KKT hits versus 4/6. Beyond that point, the stall-recovery PCG mode becomes the preferred implementation. On $120\le n\le 200$, it is 3.42$\times$ faster and improves the projected-KKT hit count to 12/15 from 9/15, with 13 objective ties and 2 large-configuration wins. On $250\le n\le 500$, the runtime ratio widens to 5.34, objective outcomes remain mostly ties (10 out of 12), and the projected-KKT hit count favors the PCG mode by 7/12 to 2/12. Only once $n$ enters the upper hundreds does the $10^{-6}$ threshold become demanding for both configurations: on $600\le n\le 1000$, neither configuration records a projected-KKT hit, but the large configuration is still 14.38$\times$ faster and ties on objective value in 8 of the 9 cases.

Supporting reduced-Hessian diagnostics reported in the Appendix show that, on this same $T=252$ benchmark family, the median positive-spectrum condition number at the equal-weight start rises from about $1.32$ at $n=4$ to $17.76$ at $n=100$, while no negative eigenvalues are observed on that block. The crossover is therefore consistent with a geometry-deterioration story: even before the ultra-scale range, the direct configuration is being asked to solve increasingly anisotropic reduced problems, which is exactly the regime in which the PCG configuration should become the preferred implementation. The resulting recommendation is therefore asymmetric: the direct configuration remains a viable precision-oriented option on the $n\le 100$ benchmark, especially around $80\le n\le 100$, but the PCG-based YAND-MVSK (large) configuration is already the preferred runtime-oriented implementation by the upper end of that grid and remains so thereafter.

Table~\ref{tab:synthetic-summary} summarizes the cross-solver evidence, and Figure~\ref{fig:synthetic-overview} reports the corresponding dimension profiles. The top row uses the $T=252$, $n\le 100$ benchmark grid. The bottom row uses the full $n>100$ large-scale segment of the same benchmark, spanning asset dimensions from 120 to 5000. Both rows overlay YAND-MVSK against Q-MVSK, UBDCA, and UDCA. Because this second-layer benchmark intentionally mixes two coefficient-certified convex profiles with one profile outside the certificate, the cross-solver comparison should be read as robustness evidence across both convex and potentially nonconvex unrestricted-MVSK instances. In Table~\ref{tab:synthetic-summary}, ``Obj.\ ties'' counts runs whose objective value lies within $10^{-6}$ of the best objective attained by any compared solver on the same instance.

\begin{table}[ht]
\centering
\scriptsize
\caption{Four-solver comparison on the synthetic benchmark at $T=252$.}
\label{tab:synthetic-summary}
\setlength{\tabcolsep}{4pt}
\resizebox{\textwidth}{!}{%
\begin{tabular}{@{}llrrrrr@{}}
\toprule
$n$ range & Method & Mean time (s) & Mean iters & Mean proj. residual & Mean gap to best & Obj. ties \\
\midrule
$n\le 100$ & YAND-MVSK (small) & $0.0711$ & $41.74$ & $2.69\times 10^{-7}$ & $1.56\times 10^{-8}$ & 54/54 \\
 & Q-MVSK & $0.0107$ & $2.704$ & $1.90\times 10^{-4}$ & $7.56\times 10^{-8}$ & 54/54 \\
 & UBDCA & $0.0304$ & $42.59$ & $1.23\times 10^{-3}$ & $3.01\times 10^{-5}$ & 43/54 \\
 & UDCA & $0.591$ & $2053.5$ & $2.37\times 10^{-2}$ & $5.38\times 10^{-3}$ & 39/54 \\
\midrule
$n>100$ & YAND-MVSK (large) & $0.332$ & $98.56$ & $3.56\times 10^{-6}$ & $1.62\times 10^{-6}$ & 34/48 \\
 & Q-MVSK & $7.685$ & $3.062$ & $2.19\times 10^{-4}$ & $3.37\times 10^{-8}$ & 47/48 \\
 & UBDCA & $0.720$ & $97.71$ & $1.37\times 10^{-1}$ & $8.72\times 10^{-2}$ & 1/48 \\
 & UDCA & $3.665$ & $2083.9$ & $1.87\times 10^{-1}$ & $1.36\times 10^{-1}$ & 0/48 \\
\bottomrule
\end{tabular}
}
\end{table}

\begin{figure*}[t]
\centering
\includegraphics[width=0.98\textwidth]{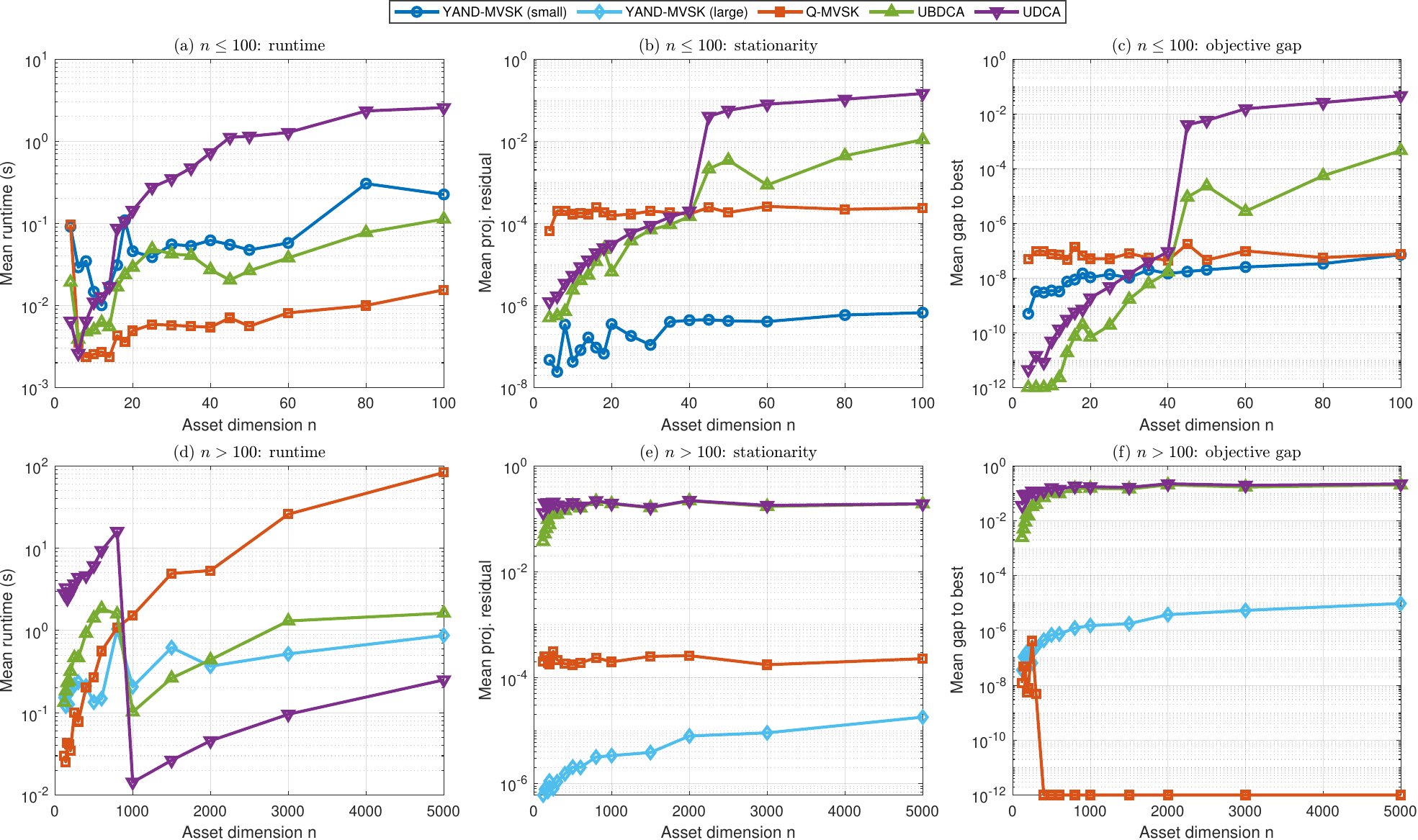}
\caption{Four-solver comparison on the synthetic benchmark at $T=252$. The top row reports the $n\le 100$ benchmark grid with YAND-MVSK (small); the bottom row reports the $n>100$ segment with YAND-MVSK (large) and stall recovery. Panels report mean runtime, mean projected-simplex KKT residual, and mean objective gap to the best method on the same instance.}
\label{fig:synthetic-overview}
\end{figure*}

\paragraph{Four-solver evidence on the benchmark grid.}
On the $n\le 100$ benchmark, Q-MVSK is the fastest method, with mean runtime $0.0107$ seconds and mean outer-iteration count 2.704 under its native step/objective stagnation safeguards. YAND-MVSK (small) matches the best observed objective on all 54 instances while keeping the mean projected-simplex residual at $2.69\times 10^{-7}$, versus $1.90\times 10^{-4}$ for Q-MVSK. Its mean runtime is $0.0711$ seconds, so it is no longer the fastest method on this grid, but it still delivers the cleanest first-order solution profile among the four methods. UBDCA is faster in raw runtime than YAND-MVSK (small), yet it ends with mean projected-simplex residual $1.23\times 10^{-3}$ and mean objective gap $3.01\times 10^{-5}$; UDCA is slower still and degrades further to $2.37\times 10^{-2}$ and $5.38\times 10^{-3}$, with mean iteration count 2053.5. The top row of Figure~\ref{fig:synthetic-overview} shows the same pattern across $n$: Q-MVSK keeps the smallest runtime profile, YAND-MVSK (small) delivers the strongest KKT profile among the four methods, and the DC baselines deteriorate as the dimension approaches the upper end of the benchmark grid. Since this block mixes two coefficient-certified convex stress profiles with one potentially nonconvex profile, that ordering is not a convex-only artifact of the benchmark design.

\paragraph{Four-solver evidence at large scale.}
On the enlarged $n>100$ segment, which extends to $n=5000$, YAND-MVSK (large) with stall recovery and Q-MVSK remain the only two methods that stay competitive on objective value. Q-MVSK retains the stronger aggregate best-observed objective statistics, with 47/48 objective ties and mean gap $3.37\times 10^{-8}$, whereas YAND-MVSK (large) records 34/48 objective ties and mean gap $1.62\times 10^{-6}$. Averaged over the full $n>100$ segment, YAND-MVSK (large) requires $0.332$ seconds, compared with $7.685$ seconds for Q-MVSK, while also delivering a much tighter mean projected-simplex residual, $3.56\times 10^{-6}$ versus $2.19\times 10^{-4}$.

A supporting overlap table in the Appendix, which isolates the common large-scale block $n\in\{120,200,400,800\}$ before the ultra-scale cases dominate the averages, already shows the same pattern: YAND-MVSK (large) and Q-MVSK are essentially tied in mean runtime ($0.265$ versus $0.260$ seconds) and almost tied on objective value, but YAND attains projected-simplex residual at most $10^{-4}$ on all 12 runs versus only 4 out of 12 for Q-MVSK. The appendix table gives the full overlap-block breakdown, and a separate sensitivity note shows that tightening Q-MVSK's legacy stagnation tolerances on representative medium- and large-scale instances raises runtime materially but improves the projected-simplex residual only modestly. Figure~\ref{fig:synthetic-overview} then shows why the full-segment averages separate so sharply: the added ultra-scale block amplifies that first-order-accuracy difference into a decisive runtime advantage. At $n=1000$, the two runtime profiles are still of the same order, but on $n\in\{1500,2000,3000,5000\}$ Q-MVSK averages about 4.9, 5.3, 25.8, and 82.9 seconds, whereas the PCG-based YAND configuration averages about 0.62, 0.37, 0.52, and 0.87 seconds. Thus, once the benchmark moves beyond the common overlap grid, the main separation is runtime and first-order accuracy rather than best-observed objective value. UBDCA and UDCA remain clearly dominated even under the common projected-KKT target and stall safeguards: their mean projected-simplex residuals remain $1.37\times 10^{-1}$ and $1.87\times 10^{-1}$, and their mean objective gaps remain $8.72\times 10^{-2}$ and $1.36\times 10^{-1}$.

\section{Economic Value of Higher Moments at Full Scale}
\label{sec:when-higher-moments-matter}

The scale result does not imply that higher moments uniformly dominate mean--variance portfolios. Its importance is methodological and economic. Before unrestricted MVSK could be solved beyond modest dimensions, the literature had limited scope for asking whether skewness and kurtosis preferences still changed outcomes materially after diversification or whether the reported gains were mainly a small-universe artifact. The scale evidence of Section~\ref{sec:computational-evidence} removes that barrier and makes a direct large-universe MV-versus-MVSK comparison feasible.

\paragraph{Real-data design.}
The real-data experiment is a full-scale economic comparison rather than a conditioning diagnostic. It uses a cleaned RESSET 5-minute A-share panel from January 2, 2020 to June 30, 2025. Missing vendor entries are filled with zero, assets with insufficient trading histories are removed, and the resulting balanced panel contains $n=5440$ stocks and $T=66412$ time periods. Thus the retained empirical panel is itself sample-rich, which matters because plug-in higher moments are easier to justify there than in an under-sampled large-universe window. We split the sample at January 1, 2024, which yields $48374$ training periods and $18038$ test periods. Covariances are estimated from demeaned training returns and expected returns from raw training means. The benchmark MV portfolio solves the exact long-only return-constrained variance problem by Gurobi. The return floor is the $q$th percentile of the in-sample asset means. The portfolios are estimated once on the training block and then held fixed over the test block, so the exercise is a static allocation comparison rather than a rolling trading strategy. The main display therefore reports the full target sweep $q\in\{0.40,0.50,0.60\}$ so that the economic role of higher moments can be read across increasingly demanding return floors rather than from a single calibration.

The MVSK comparison keeps the same long-only simplex and uses exact sample moments from the training panel. Because MVSK is preference-based rather than return-constrained, we normalize the moment coefficients by equal-weight in-sample portfolio moment magnitudes and then calibrate the mean coefficient so that the resulting MVSK portfolio lies as close as possible to the same in-sample return floor as the MV benchmark. The main text reports the kurtosis-focused profile because it gives the most direct tail-risk interpretation for the paper's economic question. The Appendix then reports side-by-side results for the skew-focused, kurtosis-focused, and balanced profiles on the same split, together with rolling-window split checks and one-time implementation-cost sensitivity for this static design. Test-period portfolio statistics are annualized using $48\times 244=11712$ 5-minute intervals per trading year.

\begin{figure*}[t]
\centering
\includegraphics[width=0.98\textwidth]{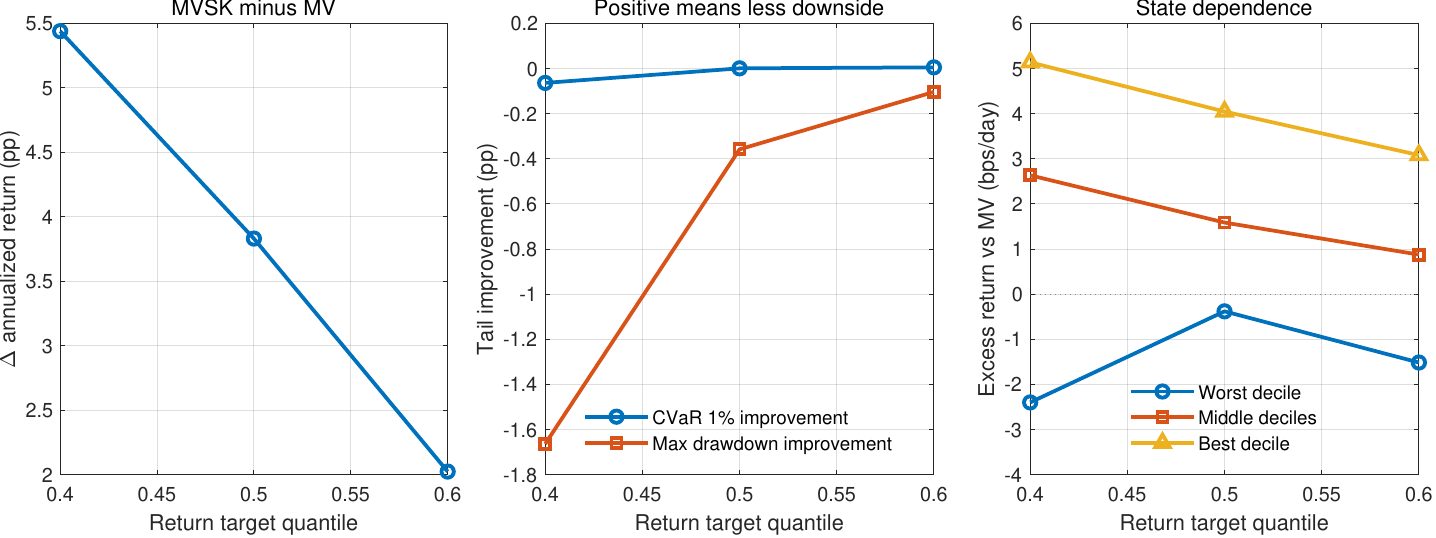}
\caption{Kurtosis-aware MVSK across return targets on the 5-minute A-share panel. Panel (a) reports annualized return gains relative to MV, panel (b) downside-risk changes, and panel (c) average daily excess returns by market-return decile.}
\label{fig:real-kurtosis-target-sweep}
\end{figure*}

\begin{table}[ht]
\centering
\scriptsize
\caption{Target sweep for kurtosis-focused MVSK on the 5-minute A-share panel. Positive $\Delta$ metrics favor MVSK.}
\label{tab:real-kurtosis-target-sweep}
\begin{tabular}{rrrrrrrrrr}
\toprule
$q$ & MV ret. & MVSK ret. & $\Delta$ ret. & $\Delta$ Sharpe & $\Delta$ CVaR$_1$ & $\Delta$ MDD & Worst & Middle & Best\\
\midrule
0.40 & 36.24 & 41.68 & 5.44 & 0.288 & 0.064 & 1.664 & -2.40 & 2.63 & 5.14\\
0.50 & 38.20 & 42.03 & 3.83 & 0.150 & -0.001 & 0.358 & -0.38 & 1.59 & 4.05\\
0.60 & 40.37 & 42.39 & 2.02 & 0.077 & -0.005 & 0.104 & -1.51 & 0.88 & 3.08\\
\bottomrule
\end{tabular}
\end{table}

\paragraph{Incremental value in the A-share panel.}
Table~\ref{tab:real-kurtosis-target-sweep} and Figure~\ref{fig:real-kurtosis-target-sweep} show that, on the baseline split of this panel, the economic value of MVSK is sharply target-dependent. At the moderate return floor $q=0.40$, kurtosis-aware MVSK raises annualized test return from $36.24\%$ to $41.68\%$ and Sharpe from $1.332$ to $1.619$, while also improving the $1\%$ CVaR by $0.064$ percentage points and reducing maximum drawdown by $1.664$ percentage points relative to exact MV. At $q=0.50$, the gains remain economically meaningful: annualized return rises by $3.83$ percentage points and Sharpe by $0.150$, while maximum drawdown still improves modestly. At the aggressive floor $q=0.60$, MVSK continues to improve annualized return by $2.02$ percentage points and Sharpe by $0.077$, but the downside advantage is largely gone: the change in $1\%$ CVaR is slightly negative and the drawdown improvement shrinks to $0.104$ percentage points.

\paragraph{Managerial implications.}
For the baseline split, the main managerial implication is that higher moments matter most when the investor's return requirement is binding but not extreme. At the moderate floor $q=0.40$, and more weakly at $q=0.50$, the fourth-moment term helps reshape a concentrated high-mean basket away from the most damaging intraday outcomes, so the investor gets more return together with some downside improvement. The portfolio-composition diagnostics point in the same direction. At $q=0.40$, the MVSK portfolio is still materially different from exact MV, with active share $11.6\%$ relative to MV and a slightly less concentrated core allocation (effective number $1.82$ versus $1.60$). Once the return floor becomes aggressive, however, the exact MV benchmark is already pushed toward a narrow set of high-mean names with substantial upside asymmetry. By $q=0.60$, the two portfolios become much closer in composition: active share falls to $8.1\%$, effective numbers are nearly identical ($2.73$ for MVSK versus $2.74$ for MV), and the top-10 weight shares are also almost the same ($67.2\%$ versus $67.5\%$). In that regime MVSK still extracts additional upside, but it no longer purchases much left-tail protection because little portfolio-design freedom remains once the return target has already driven both solutions toward the same concentrated corner. The baseline evidence therefore supports a conditional adoption rule rather than uniform replacement of MV: use MVSK when the mandate asks for more than a variance-efficient portfolio but still leaves enough design freedom for higher moments to reshape tail exposure, and do not expect much incremental downside protection once the target is so aggressive that both models are forced toward the same concentrated corner. The appendix rolling-window checks show that even this rule is market-window dependent rather than universal.

\paragraph{What MVSK is not.}
The decile panel in Figure~\ref{fig:real-kurtosis-target-sweep} also shows that, on the baseline split, MVSK is not a generic bear-market hedge. In all three target regimes, its relative performance is weakest on the worst market-return decile and strongest on the best decile, with the middle deciles also positive. The real-data gains therefore come from distributional reshaping rather than from universal crash-day outperformance. At $q=0.40$, for example, MVSK underperforms MV on the worst market decile by $2.40$ basis points per day on average, but more than compensates through gains of $2.63$ basis points in the middle deciles and $5.14$ basis points on the best decile. This pattern is economically informative: higher moments help investors redesign the payoff profile of a concentrated high-mean portfolio, but that redesign should be read as selective upside capture with reduced exposure to the most severe pathwise losses, not as a blanket insurance policy against broad market selloffs.

\paragraph{Scope and robustness.}
The real-data evidence should also be interpreted in that scope. The Appendix reports three additional checks. First, the cross-profile table shows that the target-dependence pattern on the baseline split is not unique to the kurtosis-focused profile emphasized in the main text: at $q=0.50$ the skew-focused, kurtosis-focused, and balanced profiles are nearly indistinguishable, whereas at $q=0.40$ the skew-focused and balanced profiles generate somewhat larger gains against MV. Second, rolling-window splits show that the separation from MV is state dependent. In the 2023 window the kurtosis profile remains positive across the three targets, but in the 2024 and 2025H1 windows the $q\in\{0.40,0.50\}$ allocations become almost indistinguishable from MV and the only material separation reappears at $q=0.60$, mainly through upside rather than downside protection. Third, because the portfolios are held fixed over each test block, the relevant cost check is a one-time transition from MV to MVSK rather than recurring turnover. On the baseline split, the break-even one-way implementation costs are 7416, 6110, and 3698 basis points at $q=0.40$, $0.50$, and $0.60$, so the baseline gains are not fragile to reasonable one-time trading-cost assumptions. These checks do not turn the experiment into a full trading-system evaluation, but they sharpen its interpretation: the gains are economically meaningful on the baseline split and not explained away by one-time implementation costs, yet they are market-window dependent rather than universal.

\section{Conclusion}
\label{sec:conclusion}

This paper develops a structure-exploiting YAND algorithm for unrestricted sample-moment MVSK portfolio optimization. The investor-facing model is unchanged; what changes is the numerical interface. By working directly with the return matrix, the method replaces explicit dense coskewness and cokurtosis tensors with exact sample-oracle calculations, affine-normal steps, exact quartic line search, and simplex-aware continuation. The theory supplies an exact reduced-coordinate formulation, explicit regularity constants, and a coefficient-only convexity certificate that includes the standard CRRA calibration.

The computational study shows that this reformulation changes the practical scale of unrestricted MVSK. In the controlled convex benchmark, the PCG-based configuration is faster and more accurate than Q-MVSK. On the standard small synthetic benchmark, the direct configuration remains a viable precision-oriented option, but the PCG-based configuration becomes the preferred runtime-oriented implementation by the upper end of that grid and stays competitive with the strongest unrestricted-MVSK baselines as the asset universe moves into the thousands. The economic study then shows what that scale makes possible: on a 5-minute A-share panel with 5,440 stocks, unrestricted MVSK can be compared directly with exact MV on the full universe, and the baseline split indicates that the incremental value of higher moments is strongest at moderate return targets rather than at the most aggressive ones.

The evidence therefore supports a conditional rather than universal adoption rule. On the baseline split, higher moments are most useful when the investor's return requirement is binding but still leaves room to reshape tail exposure; once the target is so aggressive that both MV and MVSK are pushed toward the same concentrated corner, the remaining gains are better read as payoff reshaping than as broad downside protection. Appendix checks also show that this economic advantage is market-window dependent rather than universal, while one-time implementation-cost adjustments are far too small to overturn the baseline gains in the static allocation design studied here. Natural next tests include rolling re-estimation with portfolio refresh, recurring trading costs, alternative target-matching rules, and regularized moment estimates.

\paragraph{Code and data availability.}
The replication package will provide code, a README file, and scripts for reproducing the synthetic experiments and rebuilding all reported tables and figures. The real-data analysis uses licensed RESSET intraday data, which cannot be redistributed. For that part of the paper, the replication materials will therefore provide data-construction code and instructions for reconstructing the sample from an authorized RESSET license.

\section*{Acknowledgments}
Y.-S. Niu's work was supported by the National Natural Science Foundation of China [Grant No.\ 42450242] and the Beijing Overseas High-Level Talent Program. A. Sheshmani's work was supported by the Beijing Natural Science Foundation [Grant No.\ BJNSF-IS24005] and the NSFC-RFIS Program [Grant No.\ W2432008]. A. Sheshmani also gratefully acknowledges support from the NSF AI Institute for Artificial Intelligence and Fundamental Interactions at the Massachusetts Institute of Technology, funded by the U.S. National Science Foundation under Cooperative Agreement PHY-2019786, as well as the Beijing Overseas High-Level Talent Program. All authors gratefully acknowledge institutional support from the Beijing Institute of Mathematical Sciences and Applications (BIMSA).

\bibliography{opre_mvsk_refs}

\clearpage
\appendix
\section*{Appendix}
\addcontentsline{toc}{section}{Appendix}

\section{Additional Derivations}

This appendix collects omitted proofs, implementation details, and supplementary numerical diagnostics.

We use the same notation as in the main paper: $R \in \R^{T \times n}$ is the sample-return matrix, $\mu$ is the sample mean, $A:=R-\bm{1}\mu^\top$ is the centered sample matrix, and $z(x):=Ax$. For a vector $u \in \R^T$, we write $u_t$ for its $t$th component and $\hadpow{u}{p}$ for its componentwise $p$th power, that is, $(\hadpow{u}{p})_t=u_t^p$. When a scalar map $g:\R\to\R$ is applied to a vector, it is understood componentwise. We use $\|\cdot\|_2$ for the Euclidean norm, $\|\cdot\|_\infty$ for the max norm, and $\|\cdot\|_{\mathrm{op}}$ for the spectral operator norm.

\subsection{Operator Structure and Curvature Decomposition}

\begin{proposition}[Exact operator representation]
\label{prop:ec-exact-operator}
For every $x \in \simplex$, the unrestricted sample-moment MVSK objective can be written exactly as
\begin{equation}
f(x)= -c_1 \mu^\top x + \frac{c_2}{T}\|z(x)\|_2^2 - \frac{c_3}{T}\sum_{t=1}^T z_t(x)^3 + \frac{c_4}{T}\sum_{t=1}^T z_t(x)^4.
\end{equation}
Moreover,
\begin{equation}
\nabla f(x)= -c_1\mu + \frac{2c_2}{T}A^\top z - \frac{3c_3}{T}A^\top \hadpow{z}{2} + \frac{4c_4}{T}A^\top \hadpow{z}{3}.
\end{equation}
and for every $v \in \R^n$,
\begin{equation}
\nabla^2 f(x)v = \frac{2c_2}{T}A^\top(Av) - \frac{6c_3}{T}A^\top(z \circ Av) + \frac{12c_4}{T}A^\top(\hadpow{z}{2} \circ Av),
\end{equation}
where $z=z(x)$.
Finally, if for directions $u,v \in \R^n$ we write
\[
\mathcal T_3(x;u,v):=D^3f(x)[u,v,\cdot] \in \R^n,
\]
then
\begin{equation}
\mathcal T_3(x;u,v)= -\frac{6c_3}{T}A^\top\!\bigl((Au)\circ(Av)\bigr) + \frac{24c_4}{T}A^\top\!\bigl(z \circ (Au)\circ(Av)\bigr),
\end{equation}
where again $z=z(x)$.
\end{proposition}

\proof{Proof.}
For each sample $t$, the centered portfolio return is $z_t(x)=a_t^\top x$. Hence the centered sample moments are
\[
m_2(x)=\frac{1}{T}\sum_{t=1}^T z_t(x)^2,
\quad
m_3(x)=\frac{1}{T}\sum_{t=1}^T z_t(x)^3,
\quad
m_4(x)=\frac{1}{T}\sum_{t=1}^T z_t(x)^4,
\]
while the sample mean is $m_1(x)=\mu^\top x$. Substituting these identities into the MVSK objective gives the displayed formula for $f(x)$.

Now define the scalar polynomial
\[
\phi(u):=c_2u^2-c_3u^3+c_4u^4.
\]
Then
\[
f(x)=-c_1\mu^\top x+\frac{1}{T}\sum_{t=1}^T \phi(z_t(x)).
\]
Because $z(x)=Ax$ is linear in $x$, the chain rule gives
\[
\nabla f(x)
= -c_1\mu+\frac{1}{T}A^\top \phi'(z),
\]
where $\phi'(u)=2c_2u-3c_3u^2+4c_4u^3$ is applied componentwise. This yields the stated gradient formula. Differentiating once more in direction $v$ gives
\[
\nabla^2 f(x)v
= \frac{1}{T}A^\top\!\bigl(\phi''(z)\circ Av\bigr),
\]
where $\phi''(u)=2c_2-6c_3u+12c_4u^2$. Expanding the componentwise product yields the displayed Hessian--vector representation.

For the third-order directional kernel, differentiate the Hessian action once more in direction $u$. Since $Av$ does not depend on $x$ and
\[
\phi'''(u)=-6c_3+24c_4u,
\]
the chain rule yields
\[
\mathcal T_3(x;u,v)
=
\frac{1}{T}A^\top\!\bigl(\phi'''(z)\circ (Au)\circ(Av)\bigr).
\]
Expanding $\phi'''(z)$ componentwise gives the displayed formula.
\endproof

\begin{proposition}[Lifted separability, curvature decomposition, and tangential conditioning]
\label{prop:ec-curvature-decomposition}
Define
\[
\psi(s):=c_2 s^2-c_3 s^3+c_4 s^4,
\qquad
\Phi(z):=\frac{1}{T}\sum_{t=1}^T \psi(z_t).
\]
Then
\[
f(x)=-c_1\mu^\top x+\Phi(Ax).
\]
Moreover,
\[
\nabla^2 f(x)=\frac{1}{T}A^\top D(x)A,
\qquad
D(x):=\operatorname{Diag}\!\bigl(\psi''(z_1(x)),\ldots,\psi''(z_T(x))\bigr),
\]
where
\[
\psi''(s)=2c_2-6c_3 s+12c_4 s^2.
\]
If $\mathcal X\subset \simplex$ is such that, for some constant $\widehat\gamma\ge 0$,
\[
\bigl|\psi''\bigl((Ax)_t\bigr)\bigr|\le \widehat\gamma
\qquad
\text{for all }x\in\mathcal X,\; t=1,\ldots,T,
\]
then for every $x\in\mathcal X$ and every tangent direction $d\in{\mathcal T}:=\{v\in\R^n:\bm{1}^\top v=0\}$,
\[
\bigl|d^\top \nabla^2 f(x)d\bigr|
\le
\frac{\widehat\gamma}{T}\|Ad\|_2^2.
\]
If, in addition, for some constants $0<\underline\gamma\le \overline\gamma$,
\[
\underline\gamma\le \psi''\bigl((Ax)_t\bigr)\le \overline\gamma
\qquad
\text{for all }x\in\mathcal X,\; t=1,\ldots,T,
\]
then for every $x\in\mathcal X$ and every tangent direction $d\in{\mathcal T}:=\{v\in\R^n:\bm{1}^\top v=0\}$,
\[
\frac{\underline\gamma}{T}\|Ad\|_2^2
\le
d^\top \nabla^2 f(x)d
\le
\frac{\overline\gamma}{T}\|Ad\|_2^2.
\]
Consequently, whenever $A$ is injective on ${\mathcal T}$,
\[
\kappa_{{\mathcal T}}\!\bigl(\nabla^2 f(x)\bigr)
\le
\frac{\overline\gamma}{\underline\gamma}\,\kappa_{{\mathcal T}}(A)^2,
\]
where
\[
\kappa_{{\mathcal T}}(A):=
\frac{\max\{\|Ad\|_2:\ d\in{\mathcal T},\ \|d\|_2=1\}}
{\min\{\|Ad\|_2:\ d\in{\mathcal T},\ \|d\|_2=1\}}
\]
and $\kappa_{{\mathcal T}}(\nabla^2 f(x))$ is defined analogously through the Rayleigh quotient on ${\mathcal T}$.
\end{proposition}

\proof{Proof.}
The identity
\[
f(x)=-c_1\mu^\top x+\Phi(Ax)
\]
is immediate from the previous proposition after renaming the scalar quartic part as $\psi$. Because $\Phi$ is separable,
\[
\nabla \Phi(z)=\frac{1}{T}\psi'(z),
\qquad
\nabla^2 \Phi(z)=\frac{1}{T}\operatorname{Diag}\!\bigl(\psi''(z_1),\ldots,\psi''(z_T)\bigr),
\]
with $\psi''(s)=2c_2-6c_3 s+12c_4 s^2$. Composing with the linear map $x\mapsto Ax$ therefore yields
\[
\nabla^2 f(x)=A^\top \nabla^2\Phi(Ax)A
=
\frac{1}{T}A^\top D(x)A.
\]

Now fix $x\in\mathcal X$ and $d\in{\mathcal T}$. Then
\[
d^\top \nabla^2 f(x)d
=
\frac{1}{T}(Ad)^\top D(x)(Ad)
=
\frac{1}{T}\sum_{t=1}^T \psi''\bigl((Ax)_t\bigr)(a_t^\top d)^2.
\]
The absolute bound on $\psi''$ immediately gives
\[
\bigl|d^\top \nabla^2 f(x)d\bigr|
\le
\frac{\widehat\gamma}{T}\sum_{t=1}^T (a_t^\top d)^2
=
\frac{\widehat\gamma}{T}\|Ad\|_2^2.
\]
If, in addition, $0<\underline\gamma\le \psi''((Ax)_t)\le \overline\gamma$, then
\[
\frac{\underline\gamma}{T}\sum_{t=1}^T (a_t^\top d)^2
\le
d^\top \nabla^2 f(x)d
\le
\frac{\overline\gamma}{T}\sum_{t=1}^T (a_t^\top d)^2,
\]
which is exactly
\[
\frac{\underline\gamma}{T}\|Ad\|_2^2
\le
d^\top \nabla^2 f(x)d
\le
\frac{\overline\gamma}{T}\|Ad\|_2^2.
\]

If $A$ is injective on ${\mathcal T}$, then the minimum in the definition of $\kappa_{{\mathcal T}}(A)$ is positive. Taking the maximum and minimum Rayleigh quotients of $\nabla^2 f(x)$ over unit vectors in ${\mathcal T}$ and using the previous sandwich bound gives
\[
\lambda_{\max}\!\bigl(\nabla^2 f(x)\vert_{\mathcal T}\bigr)
\le
\frac{\overline\gamma}{T}\max_{\substack{d\in{\mathcal T}\\ \|d\|_2=1}}\|Ad\|_2^2,
\]
\[
\lambda_{\min}\!\bigl(\nabla^2 f(x)\vert_{\mathcal T}\bigr)
\ge
\frac{\underline\gamma}{T}\min_{\substack{d\in{\mathcal T}\\ \|d\|_2=1}}\|Ad\|_2^2.
\]
Dividing the two bounds yields
\[
\kappa_{{\mathcal T}}\!\bigl(\nabla^2 f(x)\bigr)
\le
\frac{\overline\gamma}{\underline\gamma}
\frac{\max_{\|d\|_2=1,d\in{\mathcal T}}\|Ad\|_2^2}
{\min_{\|d\|_2=1,d\in{\mathcal T}}\|Ad\|_2^2}
=
\frac{\overline\gamma}{\underline\gamma}\,\kappa_{{\mathcal T}}(A)^2.
\]
\endproof

\begin{remark}[Interpretation]
\label{rem:ec-curvature-decomposition}
Proposition~\ref{prop:ec-curvature-decomposition} has two layers. First, the exact factorization $\nabla^2 f(x)=T^{-1}A^\top D(x)A$ and the absolute curvature bound show that tangential curvature magnitudes are filtered through the data map $A$ and weighted by the scalar responses $\psi''((Ax)_t)$. This decomposition remains valid on nonconvex regions: the matrix $A$ determines how strongly the return data distort Euclidean geometry, whereas $\psi$ determines the intrinsic quartic nonconvexity, sign changes, and local curvature variation of the model.

Second, on regions where $\psi''$ is uniformly positive, the same decomposition sharpens to a genuine tangential condition-number bound, separating the curvature spread $\overline\gamma/\underline\gamma$ from the geometric distortion $\kappa_{\mathcal T}(A)^2$. Outside such sign-definite regions, even if $A$ is injective on ${\mathcal T}$, mixed signs in $D(x)$ can make the tangential Hessian indefinite or singular, so no finite positive-definite condition number need exist.

This does not imply that affine-normal descent removes nonconvexity. It does explain why the affine-normal viewpoint is natural here: once the objective is written as a linear embedding followed by a separable scalar nonlinearity, the $A$-driven part of the anisotropy is exactly the part that affine-invariant geometry is intended to neutralize. This is the same mechanism emphasized in the affine-scaling results of \citet{NiuSheshmaniYau2026YAND}. There, under an invertible affine scaling model, the mapped YAND iterates and convergence constants are inherited from the unscaled base objective rather than from the artificial condition number introduced by the scaling matrix. Our MVSK decomposition is not identical to that model, but it isolates the same source of anisotropy: a linear map acts first, and the nonlinear scalar response comes afterward.

Because $\bm{1}_T^\top A=0$, one has $\operatorname{rank}(A)\le T-1$, so injectivity on ${\mathcal T}$ can occur on the full simplex only if $n\le T$. If $A$ is not injective on ${\mathcal T}$, then there exists a nonzero tangent direction $d\in{\mathcal T}$ with $Ad=0$, and therefore $\nabla^2 f(x)d=0$ for every $x$. Thus the tangential Hessian is singular on the full simplex, and the condition-number bound above should be read as meaningful only on tangent subspaces, or reduced faces, where the relevant data map is injective.
\end{remark}

\subsection{Feasible Quartic Line Search}

\begin{lemma}[Interior-safe steplength cap]
\label{lem:ec-feasibility-cap}
Fix $\tau \in [0,1/n)$ and define
\[
\simplex(\tau):=\{x \in \simplex : x_i \ge \tau \text{ for all } i\}.
\]
Let $x \in \simplex(\tau)$ and $d \in {\mathcal T}$. Define
\[
\alpha_{\max}(x,d;\tau):=
\min_{i:d_i<0}\frac{x_i-\tau}{-d_i},
\]
which is finite whenever $d\neq 0$. Then for every $\alpha \in [0,\alpha_{\max}(x,d;\tau)]$,
\[
x+\alpha d \in \simplex(\tau).
\]
\end{lemma}

\proof{Proof.}
Because $d \in {\mathcal T}$, we have $\bm{1}^\top d=0$, and therefore
\[
\bm{1}^\top(x+\alpha d)=\bm{1}^\top x=1
\]
for all $\alpha$. If $d\neq 0$ and every component of $d$ were nonnegative, then $\bm{1}^\top d=0$ would force $d=0$, a contradiction. Hence every nonzero tangent direction has at least one negative component, so $\alpha_{\max}(x,d;\tau)$ is finite. Thus only the lower bounds on the coordinates must be checked. If $d_i\ge 0$, then
\[
x_i+\alpha d_i \ge x_i \ge \tau.
\]
If $d_i<0$, then by definition of $\alpha_{\max}(x,d;\tau)$,
\[
\alpha \le \frac{x_i-\tau}{-d_i}
\qquad\Longrightarrow\qquad
x_i+\alpha d_i \ge \tau.
\]
Hence every coordinate of $x+\alpha d$ remains at least $\tau$, so $x+\alpha d \in \simplex(\tau)$.
\endproof

The next result proves the quartic exact line-search formula stated in the main paper. The key point is that along any fixed feasible direction, the unrestricted sample-moment MVSK objective reduces to a one-variable quartic polynomial whose coefficients are obtained directly from sample power sums.

\begin{proposition}[Quartic exact line search from sample power sums]
\label{prop:ec-quartic-exact-linesearch}
Fix $\tau \in [0,1/n)$, $x\in \simplex(\tau)$, and a nonzero tangent direction $d\in {\mathcal T}$. Let
\[
z:=Ax,
\qquad
w:=Ad,
\]
and define the mixed sample power sums
\[
s_{rs}(x,d):=\frac1T\sum_{t=1}^T z_t^r w_t^s,
\qquad
r,s\in \mathbb{N}\cup\{0\},
\quad
r+s\le 4.
\]
Let $\alpha_{\max}(x,d;\tau)$ be the feasibility cap from Lemma~\ref{lem:ec-feasibility-cap}. Then the line-restricted objective
\[
\varphi_{x,d}(\alpha):=f(x+\alpha d)
\]
admits the quartic representation
\[
\varphi_{x,d}(\alpha)=A_0+A_1\alpha+A_2\alpha^2+A_3\alpha^3+A_4\alpha^4,
\]
with
\[
A_0=f(x),
\qquad
A_1=-c_1\mu^\top d+2c_2 s_{11}(x,d)-3c_3 s_{21}(x,d)+4c_4 s_{31}(x,d),
\]
\[
A_2=c_2 s_{02}(x,d)-3c_3 s_{12}(x,d)+6c_4 s_{22}(x,d),
\]
\[
A_3=-c_3 s_{03}(x,d)+4c_4 s_{13}(x,d),
\qquad
A_4=c_4 s_{04}(x,d).
\]
Consequently,
\[
\varphi_{x,d}'(\alpha)=A_1+2A_2\alpha+3A_3\alpha^2+4A_4\alpha^3
\]
is cubic. An exact feasible line-search step on $\simplex(\tau)$ is therefore obtained by forming
\[
\mathcal C(x,d;\tau):=
\{0,\alpha_{\max}(x,d;\tau)\}
\cup
\{\alpha\in(0,\alpha_{\max}(x,d;\tau)):\varphi_{x,d}'(\alpha)=0\},
\]
and choosing any
\[
\alpha_\star \in \arg\min_{\alpha\in \mathcal C(x,d;\tau)} \varphi_{x,d}(\alpha).
\]
If $z=Ax$ is cached, the coefficients $A_1,\ldots,A_4$ can be assembled with one additional sample projection $w=Ad$ and $O(T)$ scalar operations, hence in $O(Tn)$ arithmetic operations and $O(T)$ working memory before the final cubic-root computation.
\end{proposition}

\proof{Proof.}
For every $\alpha\in\R$,
\[
A(x+\alpha d)=Ax+\alpha Ad=z+\alpha w.
\]
Substituting $x+\alpha d$ into the exact operator representation from Proposition~\ref{prop:ec-exact-operator} therefore gives
\[
\varphi_{x,d}(\alpha)
=
-c_1\mu^\top(x+\alpha d)
+\frac{c_2}{T}\sum_{t=1}^T (z_t+\alpha w_t)^2
-\frac{c_3}{T}\sum_{t=1}^T (z_t+\alpha w_t)^3
+\frac{c_4}{T}\sum_{t=1}^T (z_t+\alpha w_t)^4.
\]
Apply the binomial identities
\[
(z_t+\alpha w_t)^2=z_t^2+2\alpha z_t w_t+\alpha^2 w_t^2,
\]
\[
(z_t+\alpha w_t)^3=z_t^3+3\alpha z_t^2w_t+3\alpha^2 z_t w_t^2+\alpha^3 w_t^3,
\]
\[
(z_t+\alpha w_t)^4=z_t^4+4\alpha z_t^3w_t+6\alpha^2 z_t^2w_t^2+4\alpha^3 z_t w_t^3+\alpha^4 w_t^4.
\]
After summing over $t$ and using the definition of $s_{rs}(x,d)$, the coefficients of like powers of $\alpha$ are exactly
\[
A_0
=
-c_1\mu^\top x
+c_2 s_{20}(x,d)
-c_3 s_{30}(x,d)
+c_4 s_{40}(x,d)
=
f(x),
\]
\[
A_1=-c_1\mu^\top d+2c_2 s_{11}(x,d)-3c_3 s_{21}(x,d)+4c_4 s_{31}(x,d),
\]
\[
A_2=c_2 s_{02}(x,d)-3c_3 s_{12}(x,d)+6c_4 s_{22}(x,d),
\]
\[
A_3=-c_3 s_{03}(x,d)+4c_4 s_{13}(x,d),
\qquad
A_4=c_4 s_{04}(x,d).
\]
This proves the quartic representation of $\varphi_{x,d}$. Differentiating term by term gives
\[
\varphi_{x,d}'(\alpha)=A_1+2A_2\alpha+3A_3\alpha^2+4A_4\alpha^3,
\]
which is cubic.

By Lemma~\ref{lem:ec-feasibility-cap}, every $\alpha\in[0,\alpha_{\max}(x,d;\tau)]$ yields a feasible portfolio $x+\alpha d\in \simplex(\tau)$. Since $\varphi_{x,d}$ is continuous on the closed interval $[0,\alpha_{\max}(x,d;\tau)]$ and differentiable on its interior, any minimizer on that interval must be either an endpoint or an interior stationary point. Hence every exact feasible minimizer belongs to $\mathcal C(x,d;\tau)$, and minimizing over that finite candidate set yields an exact feasible line-search step.

Finally, if $z=Ax$ is already cached at the current iterate, then forming $w=Ad$ requires one matrix--vector multiplication by $A$, hence $O(Tn)$ arithmetic operations. The scalar sums $s_{rs}(x,d)$ are then obtained from componentwise products of $z$ and $w$ in $O(T)$ arithmetic operations and $O(T)$ working memory. The remaining cubic-root calculation is constant-size and therefore does not alter the leading-order complexity claim.
\endproof

\begin{remark}[Focus of the current paper]
\label{rem:ec-focus}
The main paper focuses on the exact sample-oracle YAND solver because that is already sufficient to remove the dense quartic bottleneck and to make unrestricted MVSK operational at materially larger scales. Accordingly, this appendix records the exact-oracle proofs and the YAND-specific implementation details first.
\end{remark}

\section{Proofs Omitted from the Main Paper}

\subsection{Stationarity and Reduced-Coordinate Representation}

\begin{definition}[Tangent-gradient residual on the simplex interior]
\label{def:ec-stationarity}
Let
\[
{\mathcal T}:=\{v \in \R^n : \bm{1}^{\top}v=0\},
\qquad
P_{\mathcal T}:=I_n-\frac{1}{n}\bm{1}\bm{1}^{\top}.
\]
For $x \in \operatorname{ri}(\simplex)$, define
\[
G_{\mathcal T}(x):=P_{\mathcal T}\nabla f(x).
\]
\end{definition}

\begin{proposition}[Interior stationarity criterion]
\label{prop:ec-interior-stationarity}
For $x \in \operatorname{ri}(\simplex)$,
\[
G_{\mathcal T}(x)=0
\qquad\Longleftrightarrow\qquad
\begin{gathered}
x \text{ is a first-order stationary point of the}\\
\text{simplex-constrained MVSK problem.}
\end{gathered}
\]
\end{proposition}

\proof{Proof.}
Because $x \in \operatorname{ri}(\simplex)$, the only active simplex constraint is $\bm{1}^\top x=1$. Therefore the KKT condition for the simplex-constrained MVSK problem at $x$ is
\[
\nabla f(x)+\lambda \bm{1}=0
\]
for some scalar $\lambda$, which is equivalent to $\nabla f(x)$ being orthogonal to ${\mathcal T}$. This is the same as $P_{\mathcal T}\nabla f(x)=0$, that is, $G_{\mathcal T}(x)=0$.
\endproof

\begin{proposition}[Affine-subspace reduction]
\label{prop:ec-simplex-reduction}
Fix $x^{\mathrm{ref}} \in \operatorname{ri}(\simplex)$ and let $U \in \R^{n \times (n-1)}$ have orthonormal columns spanning ${\mathcal T}$. Define
\[
\mathcal Y:=\{y \in \R^{n-1} : x^{\mathrm{ref}}+Uy \in \operatorname{ri}(\simplex)\},
\qquad
\phi(y):=f(x^{\mathrm{ref}}+Uy).
\]
Then, for every $x=x^{\mathrm{ref}}+Uy \in \operatorname{ri}(\simplex)$,
\[
\nabla \phi(y)=U^\top \nabla f(x),
\qquad
\nabla^2 \phi(y)=U^\top \nabla^2 f(x)U.
\]
Moreover,
\[
G_{\mathcal T}(x)=0
\qquad\Longleftrightarrow\qquad
\nabla \phi(y)=0.
\]
\end{proposition}
\proof{Proof.}
The gradient and Hessian identities are immediate from the chain rule, because the map $y \mapsto x^{\mathrm{ref}}+Uy$ is affine with Jacobian $U$. Thus
\[
\nabla \phi(y)=U^\top \nabla f(x),
\qquad
\nabla^2 \phi(y)=U^\top \nabla^2 f(x)U.
\]
Since the columns of $U$ form an orthonormal basis of ${\mathcal T}$, we have
\[
UU^\top=P_{\mathcal T}.
\]
Hence
\[
G_{\mathcal T}(x)=0
\quad\Longleftrightarrow\quad
UU^\top \nabla f(x)=0
\quad\Longleftrightarrow\quad
U^\top \nabla f(x)=0
\quad\Longleftrightarrow\quad
\nabla \phi(y)=0.
\]
\endproof

\begin{proposition}[Reduced-coordinate exact oracle]
\label{prop:ec-reduced-oracle}
Fix $x^{\mathrm{ref}} \in \operatorname{ri}(\simplex)$ and let $U \in \R^{n \times (n-1)}$ have orthonormal columns spanning ${\mathcal T}:=\{v \in \R^n : \bm{1}^\top v = 0\}$. Define
\[
\phi(y):=f(x^{\mathrm{ref}}+Uy),
\qquad
x=x^{\mathrm{ref}}+Uy.
\]
Then for every $y \in \R^{n-1}$ and directions $u,v \in \R^{n-1}$,
\[
\nabla \phi(y)=U^\top \nabla f(x),
\qquad
\nabla^2 \phi(y)u=U^\top \nabla^2 f(x)(Uu),
\]
\[
\mathcal T_{3,\phi}(y;u,v):=D^3\phi(y)[u,v,\cdot]
=
U^\top \mathcal T_3(x;Uu,Uv).
\]
If $A$ is stored explicitly and the tangent basis $U$ is represented densely, one evaluation of $\phi(y)$, $\nabla \phi(y)$, one reduced Hessian--vector action, and one reduced directional third-order action each requires $O(Tn+n^2)$ arithmetic operations and no dense high-order moment tensor.
\end{proposition}

\proof{Proof.}
The value identity is immediate from the definition of $\phi$. Because the map $y \mapsto x^{\mathrm{ref}}+Uy$ is affine with Jacobian $U$, the chain rule gives
\[
\nabla \phi(y)=U^\top \nabla f(x),
\qquad
\nabla^2 \phi(y)=U^\top \nabla^2 f(x)U.
\]
Applying the Hessian to a direction $u$ yields the displayed reduced Hessian--vector identity. Differentiating once more in direction $v$ and using linearity of the affine map gives
\[
D^3\phi(y)[u,v,\cdot]
=
U^\top D^3 f(x)[Uu,Uv,\cdot]
=
U^\top \mathcal T_3(x;Uu,Uv).
\]
The complexity claim follows by combining these identities with the exact sample-oracle formulas from Proposition~\ref{prop:ec-exact-operator}. Each reduced oracle query needs matrix--vector multiplies with $A$ or $A^\top$, basis multiplies with $U$ or $U^\top$, and elementwise vector operations. For a dense basis $U$, the sample-oracle work contributes $O(Tn)$ arithmetic operations and the basis changes contribute $O(n^2)$, giving a total of $O(Tn+n^2)$. The statement does not rely on $T$ dominating $n$; if a more structured tangent basis is used, the $O(n^2)$ term can be reduced accordingly.
\endproof

\subsection{Regularity, Convexity Certificates, Transfer, and Structural PL}

\begin{proposition}[Explicit MVSK regularity constants on an interior simplex slice]
\label{prop:ec-explicit-constants}
Fix $\tau \in (0,1/n)$ and define
\[
B_\tau:=\sup_{x\in\simplex(\tau)}\|Ax\|_\infty.
\]
Then for every $x,y\in\simplex(\tau)$,
\[
\|\nabla^2 f(x)\|_{\mathrm{op}}
\le
L_\tau:=
\frac{\|A\|_{\mathrm{op}}^2}{T}
\bigl(2c_2+6c_3B_\tau+12c_4B_\tau^2\bigr),
\]
and
\[
\|\nabla^2 f(x)-\nabla^2 f(y)\|_{\mathrm{op}}
\le
M_\tau\|x-y\|_2,
\qquad
M_\tau:=
\frac{\|A\|_{\mathrm{op}}^3}{T}
\bigl(6c_3+24c_4B_\tau\bigr).
\]
In particular, $f$ is $L_\tau$-smooth and has $M_\tau$-Lipschitz Hessian on $\simplex(\tau)$. Moreover,
\[
B_\tau\le \|A\|_{\mathrm{op}},
\]
so one may use the coarser bounds
\[
L_\tau
\le
\frac{\|A\|_{\mathrm{op}}^2}{T}
\bigl(2c_2+6c_3\|A\|_{\mathrm{op}}+12c_4\|A\|_{\mathrm{op}}^2\bigr),
\]
\[
M_\tau
\le
\frac{\|A\|_{\mathrm{op}}^3}{T}
\bigl(6c_3+24c_4\|A\|_{\mathrm{op}}\bigr).
\]
\end{proposition}

\proof{Proof.}
Recall from Proposition~\ref{prop:ec-curvature-decomposition} that
\[
\nabla^2 f(x)=\frac{1}{T}A^\top D(x)A,
\qquad
D(x)=\operatorname{Diag}\!\bigl(\psi''((Ax)_1),\ldots,\psi''((Ax)_T)\bigr).
\]
Therefore
\[
\|\nabla^2 f(x)\|_{\mathrm{op}}
\le
\frac{1}{T}\|A\|_{\mathrm{op}}^2\|D(x)\|_{\mathrm{op}}
=
\frac{1}{T}\|A\|_{\mathrm{op}}^2
\max_{1\le t\le T}\bigl|\psi''((Ax)_t)\bigr|.
\]
Because $|(Ax)_t|\le \|Ax\|_\infty\le B_\tau$ on $\simplex(\tau)$ and
\[
\psi''(s)=2c_2-6c_3s+12c_4s^2,
\]
we obtain
\[
\max_{|s|\le B_\tau}|\psi''(s)|
\le
2c_2+6c_3B_\tau+12c_4B_\tau^2,
\]
which yields the stated bound for $\|\nabla^2 f(x)\|_{\mathrm{op}}$.

For the Hessian difference,
\[
\nabla^2 f(x)-\nabla^2 f(y)
=
\frac{1}{T}A^\top\!\bigl(D(x)-D(y)\bigr)A,
\]
so
\[
\|\nabla^2 f(x)-\nabla^2 f(y)\|_{\mathrm{op}}
\le
\frac{1}{T}\|A\|_{\mathrm{op}}^2\|D(x)-D(y)\|_{\mathrm{op}}.
\]
Since $x,y\in\simplex(\tau)$, every component of $Ax$ and $Ay$ lies in $[-B_\tau,B_\tau]$. The mean-value theorem applied componentwise to $\psi''$ therefore gives
\[
\|D(x)-D(y)\|_{\mathrm{op}}
\le
\max_{|s|\le B_\tau}|\psi'''(s)|\,\|A(x-y)\|_\infty.
\]
Because
\[
\psi'''(s)=-6c_3+24c_4s,
\]
we have
\[
\max_{|s|\le B_\tau}|\psi'''(s)|
\le
6c_3+24c_4B_\tau.
\]
Also,
\[
\|A(x-y)\|_\infty
\le
\|A(x-y)\|_2
\le
\|A\|_{\mathrm{op}}\|x-y\|_2.
\]
Combining the last three displays yields the claimed Hessian-Lipschitz bound.

Finally, for every $x\in\simplex(\tau)\subset\simplex$,
\[
\|Ax\|_\infty\le \|Ax\|_2\le \|A\|_{\mathrm{op}}\|x\|_2\le \|A\|_{\mathrm{op}},
\]
because $\|x\|_2\le \|x\|_1=1$ on the simplex. Hence $B_\tau\le \|A\|_{\mathrm{op}}$, and substituting this inequality into the previous formulas gives the coarse bounds.
\endproof

\begin{proposition}[Automatic MVSK regularity on an interior simplex slice]
\label{prop:ec-mvsk-regularity}
Fix $\tau \in (0,1/n)$ and define
\[
\simplex(\tau):=\{x \in \simplex : x_i \ge \tau \text{ for all } i\}.
\]
Then $\simplex(\tau)$ is a nonempty compact subset of $\operatorname{ri}(\simplex)$. Moreover, the sample-moment MVSK objective $f$ is a polynomial of degree at most four on $\R^n$. Consequently, there exist an open neighborhood $\mathcal N_\tau$ of $\simplex(\tau)$ and a constant $L_\tau>0$ such that $f$ is $L_\tau$-smooth on $\mathcal N_\tau$ and bounded below on $\mathcal N_\tau$. In particular, if an exact-oracle YAND implementation is safeguarded so that all iterates and accepted Armijo steps remain in $\simplex(\tau)$, then items~1 and~3 of Theorem~\ref{thm:ec-yand-stationarity} hold with $K=\simplex(\tau)$.
\end{proposition}

\proof{Proof.}
Because $\tau \in (0,1/n)$, the uniform portfolio $(1/n)\bm{1}$ belongs to $\simplex(\tau)$, so the set is nonempty. It is closed and bounded in $\R^n$, hence compact, and the strict inequality $x_i\ge\tau>0$ implies $\simplex(\tau)\subset \operatorname{ri}(\simplex)$.

The operator representation established earlier in the appendix shows that $f$ is a linear combination of a linear term, a quadratic term, a cubic term, and a quartic term in $x$. Hence $f$ is a polynomial of degree at most four on $\R^n$, and in particular $f\in C^\infty(\R^n)$.

Since $\simplex(\tau)$ is compact and contained in the open set $\operatorname{ri}(\simplex)$, there exists $\delta>0$ such that the closed Euclidean $\delta$-neighborhood
\[
\overline{\mathcal N}_\tau:=\{u \in \R^n : \operatorname{dist}(u,\simplex(\tau))\le \delta\}
\]
is compact. Because $\nabla^2 f$ is continuous on $\R^n$, the quantity
\[
L_\tau:=\sup_{u \in \overline{\mathcal N}_\tau}\|\nabla^2 f(u)\|_{\mathrm{op}}
\]
is finite. The mean-value theorem then implies that $f$ is $L_\tau$-smooth on the interior $\mathcal N_\tau$ of $\overline{\mathcal N}_\tau$. Since $f$ is continuous on the compact set $\overline{\mathcal N}_\tau$, it is also bounded below there, and therefore on $\mathcal N_\tau$.

Finally, if the implementation accepts only iterates in $\simplex(\tau)$, then all iterates and accepted Armijo steps lie in the compact set $K=\simplex(\tau)$. The previous paragraph verifies item~1 of Theorem~\ref{thm:ec-yand-stationarity}, and the safeguard itself gives item~3.
\endproof

\begin{remark}[What is automatic and what is imported]
\label{rem:ec-transfer}
Proposition~\ref{prop:ec-mvsk-regularity} shows that the MVSK model itself supplies the existence of the smoothness and boundedness constants needed in Theorem~\ref{thm:ec-yand-stationarity} once the algorithm is kept on an interior simplex slice, and Proposition~\ref{prop:ec-explicit-constants} makes the corresponding constants explicit in terms of $\|A\|_{\mathrm{op}}$ and $c$. The only additional regularity assumption not generated automatically by the MVSK model is the bounded-direction condition on the exact reduced-coordinate YAND direction imported from \citet{NiuSheshmaniYau2026YAND}. The proof below uses the corresponding angle bound and Armijo decrease estimates, followed by the global stationarity and PL conclusions in that paper. The PL condition remains an optional strengthening in general, but Corollary~\ref{cor:ec-structural-pl} later in this subsection identifies a concrete MVSK regime in which it holds automatically.
\end{remark}

\begin{proposition}[A sufficient condition for the imported bounded-direction hypothesis]
\label{prop:ec-tbounded-bridge}
Let $m:=n-1$ and consider the reduced objective $\phi:\mathcal Y\subset\R^m\to\R$. For any $y\in\mathcal Y$ with $\nabla\phi(y)\neq 0$, define the normalized reduced gradient
\[
\nu(y):=\frac{\nabla\phi(y)}{\|\nabla\phi(y)\|_2}.
\]
Let $Q_y\in\R^{m\times(m-1)}$ denote the orthonormal tangent-basis operator induced by the Householder frame used in the exact matrix-free log-determinant YAND implementation, so that $Q_y^\top Q_y=I_{m-1}$ and $Q_y^\top \nu(y)=0$. Define
\[
h(y):=Q_y^\top \nabla^2\phi(y)\nu(y),
\]
the regularized tangent operator
\[
\mathcal H_{T,\lambda}(y)\eta:=Q_y^\top \nabla^2\phi(y)(Q_y\eta)+\lambda \eta,
\qquad
\eta\in\R^{m-1},
\]
and let $a(y)\in\R^{m-1}$ denote the exact log-determinant correction term returned by the exact trace mode. Finally, define
\[
u(y):=\mathcal H_{T,\lambda}(y)^{-1}\Bigl(h(y)-\frac{\|\nabla\phi(y)\|_2}{m+1}a(y)\Bigr),
\qquad
\widetilde d_y:=Q_yu(y)-\nu(y).
\]
Suppose there exist finite constants $G,H,A,M$ such that for all $y\in\mathcal Y_K$ with $\nabla\phi(y)\neq 0$,
\[
\|\nabla\phi(y)\|_2\le G,\qquad
\|h(y)\|_2\le H,\qquad
\|a(y)\|_2\le A,\qquad
\|\mathcal H_{T,\lambda}(y)^{-1}\|_{\mathrm{op}}\le M.
\]
Then
\[
\|u(y)\|_2\le \beta_{\mathrm{dir}}:=M\Bigl(H+\frac{GA}{m+1}\Bigr)
\]
for all such $y$. In particular, the unnormalized exact reduced-coordinate log-determinant YAND direction $\widetilde d_y$ has uniformly bounded tangential magnitude on $\mathcal Y_K$, so the imported bounded-direction hypothesis of \citet[Assumption~7.1]{NiuSheshmaniYau2026YAND} holds there with constant $\beta_{\mathrm{dir}}$. If the implementation rescales $\widetilde d_y$ by a positive scalar, the angle with $-\nabla\phi(y)$ is unchanged, so the angle conclusion used in Theorem~\ref{thm:ec-yand-stationarity} is preserved.
\end{proposition}

\proof{Proof.}
The exact matrix-free log-determinant YAND routine computes the tangent variable $u(y)$ by solving the regularized tangent linear system
\[
\mathcal H_{T,\lambda}(y)u(y)=h(y)-\frac{\|\nabla\phi(y)\|_2}{m+1}a(y).
\]
Therefore
\[
\|u(y)\|_2
\le
\|\mathcal H_{T,\lambda}(y)^{-1}\|_{\mathrm{op}}
\left(\|h(y)\|_2+\frac{\|\nabla\phi(y)\|_2}{m+1}\|a(y)\|_2\right)
\le
M\Bigl(H+\frac{GA}{m+1}\Bigr).
\]
Because $Q_y$ is an orthonormal tangent-basis operator, it preserves Euclidean norms and lifts $u(y)$ into the tangent space orthogonal to $\nu(y)$. Hence, in the normal-aligned frame used by YAND, the direction $\widetilde d_y=Q_yu(y)-\nu(y)$ has tangential component $u(y)$, and the previous bound gives a uniform bound on the tangential magnitude. This yields the imported $T$-boundedness condition. If the implementation rescales $\widetilde d_y$ by a positive scalar, the resulting direction has the same angle with $-\nabla\phi(y)$, so the angle-bound conclusion of \citet[Lemma~7.3]{NiuSheshmaniYau2026YAND} still applies.
\endproof

\begin{theorem}[YAND transfer on the simplex interior]
\label{thm:ec-yand-stationarity}
Let $\{x^k\}\subset \operatorname{ri}(\simplex)$ be generated by feasible updates
\[
x^{k+1}=x^k+\alpha_k U d_y^k,
\]
where $U$ is as in Proposition~\ref{prop:ec-simplex-reduction}, $d_y^k$ is the reduced-coordinate YAND direction for $\phi$, and $\alpha_k$ is chosen by Armijo backtracking with parameter $\sigma \in (0,1)$. Assume:
\begin{enumerate}
\item there exists a compact set $K\subset \operatorname{ri}(\simplex)$ containing all iterates, and $f$ is $L$-smooth and bounded below on a neighborhood of $K$;
\item the reduced-coordinate directions satisfy the bounded-direction condition of \citet[Assumption~7.1]{NiuSheshmaniYau2026YAND} with constant $\beta_{\mathrm{dir}}<\infty$;
\item the accepted Armijo steps remain in $K$.
\end{enumerate}
Then either the method terminates at an interior first-order stationary point of the simplex-constrained MVSK problem, or
\[
\lim_{k\to\infty}\|G_{\mathcal T}(x^k)\|_2 = 0.
\]
If $\{x^k\}$ has an accumulation point $\bar x \in K$, then $\bar x$ is a first-order stationary point. If, in addition, $\phi$ satisfies the Polyak--Lojasiewicz inequality on
\[
\mathcal Y_K:=\{y \in \R^{n-1}: x^{\mathrm{ref}}+Uy \in K\}
\]
namely,
\[
\frac12\|\nabla \phi(y)\|_2^2
\ge
\mu_{\mathrm{PL}}\bigl(\phi(y)-\phi^\star\bigr)
\qquad
\text{for all } y\in\mathcal Y_K,
\]
with constant $\mu_{\mathrm{PL}}>0$ and $\phi^\star:=\inf_{y\in \mathcal Y_K}\phi(y)$, then
\[
f(x^{k+1})-f^\star
\le
\bigl(1-\rho_{\mathrm{PL}}\bigr)\bigl(f(x^k)-f^\star\bigr),
\qquad
\rho_{\mathrm{PL}}=
\frac{2\sigma(1-\sigma)\mu_{\mathrm{PL}}}{L(1+\beta_{\mathrm{dir}}^2)},
\]
where $f^\star:=\inf_{x \in K} f(x)$.
Moreover, if $\phi$ satisfies the local nondegeneracy and step-size hypotheses of \citet[Assumption~7.14, Theorem~7.16, and Remark~7.17]{NiuSheshmaniYau2026YAND} at an interior minimizer $\bar y\in\mathcal Y_K$, then the reduced iterates converge locally quadratically to $\bar y$, and therefore $x^k=x^{\mathrm{ref}}+Uy^k$ converges locally quadratically to $\bar x:=x^{\mathrm{ref}}+U\bar y$.
\end{theorem}

\proof{Proof.}
Fix $x^{\mathrm{ref}} \in K$ and let $y^k$ be the reduced coordinates defined by
\[
x^k=x^{\mathrm{ref}}+Uy^k.
\]
Because each step is feasible and tangent to the affine hull of the simplex, the reduced iterates satisfy
\[
y^{k+1}=y^k+\alpha_k d_y^k.
\]
Define $\phi(y):=f(x^{\mathrm{ref}}+Uy)$ on $\mathcal Y_K$. Since $U$ is orthonormal and $f$ is $L$-smooth on a neighborhood of $K$, the reduced objective $\phi$ is also $L$-smooth on $\mathcal Y_K$. It is bounded below there because $f$ is bounded below on $K$.

By assumption, the reduced-coordinate YAND directions satisfy the bounded-direction condition of \citet[Assumption~7.1]{NiuSheshmaniYau2026YAND} with constant $\beta_{\mathrm{dir}}$. Therefore the angle bound and Armijo decrease estimates from \citet[Lemma~7.3 and Lemma~7.4]{NiuSheshmaniYau2026YAND} apply to the sequence $\{y^k\}$. Applying \citet[Theorem~7.5]{NiuSheshmaniYau2026YAND} to the reduced problem yields either finite termination at $\nabla \phi(y^k)=0$ or
\[
\lim_{k\to\infty}\|\nabla \phi(y^k)\|_2=0.
\]
By Proposition~\ref{prop:ec-simplex-reduction},
\[
\nabla \phi(y^k)=U^\top \nabla f(x^k),
\qquad
G_{\mathcal T}(x^k)=UU^\top \nabla f(x^k).
\]
Since $U$ has orthonormal columns,
\[
\|G_{\mathcal T}(x^k)\|_2
=
\|UU^\top \nabla f(x^k)\|_2
=
\|U^\top \nabla f(x^k)\|_2
=
\|\nabla \phi(y^k)\|_2,
\]
which proves the claimed tangent-gradient convergence.

Now let $\bar x \in K$ be an accumulation point of $\{x^k\}$. Passing to a convergent subsequence $x^{k_j}\to \bar x$ and using continuity of $\nabla f$, we obtain
\[
G_{\mathcal T}(\bar x)=\lim_{j\to\infty}G_{\mathcal T}(x^{k_j})=0.
\]
Proposition~\ref{prop:ec-interior-stationarity} then implies that $\bar x$ is a first-order stationary point of the simplex-constrained MVSK problem.

Finally, if $\phi$ satisfies the Polyak--Lojasiewicz inequality on $\mathcal Y_K$, then \citet[Corollary~7.8]{NiuSheshmaniYau2026YAND} gives
\[
\phi(y^{k+1})-\phi^\star
\le
\bigl(1-\rho_{\mathrm{PL}}\bigr)\bigl(\phi(y^k)-\phi^\star\bigr),
\qquad
\rho_{\mathrm{PL}}=
\frac{2\sigma(1-\sigma)\mu_{\mathrm{PL}}}{L(1+\beta_{\mathrm{dir}}^2)},
\]
where $\phi^\star:=\inf_{y\in \mathcal Y_K}\phi(y)$. Because $\phi(y)=f(x^{\mathrm{ref}}+Uy)$, we have $\phi(y^k)=f(x^k)$ and $\phi^\star=f^\star$, which proves the displayed linear rate.

For the local statement, apply \citet[Theorem~7.16 and Remark~7.17]{NiuSheshmaniYau2026YAND} to the reduced problem $\phi$ near $\bar y$. The map $y\mapsto x^{\mathrm{ref}}+Uy$ is affine with orthonormal linear part $U$, so it preserves local neighborhoods and Euclidean distances:
\[
\|x^k-\bar x\|_2=\|U(y^k-\bar y)\|_2=\|y^k-\bar y\|_2.
\]
Hence the quadratic recursion for $\|y^k-\bar y\|_2$ transfers verbatim to $\|x^k-\bar x\|_2$.
\endproof

\begin{corollary}[A low-dimensional structural PL regime for reduced MVSK]
\label{cor:ec-structural-pl}
Fix $\tau \in (0,1/n)$ and $x^{\mathrm{ref}}\in\simplex(\tau)$. Let $U\in\R^{n\times(n-1)}$ have orthonormal columns spanning ${\mathcal T}$ and define
\[
\mathcal Y_\tau:=\{y\in\R^{n-1}:x^{\mathrm{ref}}+Uy\in\simplex(\tau)\},
\qquad
\phi(y):=f(x^{\mathrm{ref}}+Uy).
\]
Let
\[
\underline\gamma_\tau:=\min_{|s|\le B_\tau}\psi''(s),
\qquad
\overline\gamma_\tau:=\max_{|s|\le B_\tau}\psi''(s).
\]
If $\underline\gamma_\tau>0$, $n\le T$, and $A$ is injective on ${\mathcal T}$, then for every $y\in\mathcal Y_\tau$,
\[
\mu_\tau I_{n-1}
\preceq
\nabla^2\phi(y)
\preceq
L_{\tau,\phi}I_{n-1},
\]
where
\[
\mu_\tau:=\frac{\underline\gamma_\tau}{T}\sigma_{\min}(AU)^2,
\qquad
L_{\tau,\phi}:=\frac{\overline\gamma_\tau}{T}\sigma_{\max}(AU)^2.
\]
Hence $\phi$ is $\mu_\tau$-strongly convex on $\mathcal Y_\tau$ and satisfies the Polyak--Lojasiewicz inequality there with the same constant. Under the assumptions of Theorem~\ref{thm:ec-yand-stationarity}, the exact-oracle YAND iterates therefore obey
\[
f(x^{k+1})-f^\star
\le
\bigl(1-\rho_\tau\bigr)\bigl(f(x^k)-f^\star\bigr),
\qquad
\rho_\tau=
\frac{2\sigma(1-\sigma)\mu_\tau}{L_{\tau,\phi}(1+\beta_{\mathrm{dir}}^2)}.
\]
\end{corollary}

\proof{Proof.}
Fix $y\in\mathcal Y_\tau$ and set $x:=x^{\mathrm{ref}}+Uy\in\simplex(\tau)$. For any $u\in\R^{n-1}$, let $d:=Uu\in{\mathcal T}$. By Proposition~\ref{prop:ec-simplex-reduction},
\[
u^\top \nabla^2\phi(y)u=d^\top \nabla^2 f(x)d.
\]
Applying Proposition~\ref{prop:ec-curvature-decomposition} on $\simplex(\tau)$ gives
\[
\frac{\underline\gamma_\tau}{T}\|Ad\|_2^2
\le
u^\top \nabla^2\phi(y)u
\le
\frac{\overline\gamma_\tau}{T}\|Ad\|_2^2.
\]
Because $d=Uu$ and $U$ has orthonormal columns,
\[
\|Ad\|_2=\|AUu\|_2.
\]
If $A$ is injective on ${\mathcal T}$, then $AU$ has full column rank, so
\[
\sigma_{\min}(AU)\|u\|_2
\le
\|AUu\|_2
\le
\sigma_{\max}(AU)\|u\|_2.
\]
Combining the last three displays yields
\[
\frac{\underline\gamma_\tau}{T}\sigma_{\min}(AU)^2\|u\|_2^2
\le
u^\top \nabla^2\phi(y)u
\le
\frac{\overline\gamma_\tau}{T}\sigma_{\max}(AU)^2\|u\|_2^2,
\]
which is exactly the matrix inequality with constants $\mu_\tau$ and $L_{\tau,\phi}$.

The set $\mathcal Y_\tau$ is convex because it is the affine preimage of the convex set $\simplex(\tau)$. Hence the Hessian lower bound implies that $\phi$ is $\mu_\tau$-strongly convex on $\mathcal Y_\tau$, while the upper bound implies $L_{\tau,\phi}$-smoothness there. Strong convexity implies the Polyak--Lojasiewicz inequality with the same constant $\mu_\tau$, so the linear decrease estimate follows from Theorem~\ref{thm:ec-yand-stationarity} after substituting $\mu_{\mathrm{PL}}=\mu_\tau$ and $L=L_{\tau,\phi}$.
\endproof

\begin{corollary}[A coefficient-only convexity certificate for MVSK]
\label{cor:ec-convexity}
If
\[
c_4>0,
\qquad
8c_2c_4>3c_3^2,
\]
then
\[
\underline\gamma_\tau
\ge
2c_2-\frac{3c_3^2}{4c_4}
>
0
\]
for every $\tau\in(0,1/n)$. Consequently, $\psi$ is globally strictly convex on $\R$, $\Phi(z)=T^{-1}\sum_{t=1}^T\psi(z_t)$ is convex on $\R^T$, and
\[
f(x)=-c_1\mu^\top x+\Phi(Ax)
\]
is convex on $\simplex$.
\end{corollary}

\proof{Proof.}
Completing the square gives
\[
\psi''(s)=12c_4\Bigl(s-\frac{c_3}{4c_4}\Bigr)^2+2c_2-\frac{3c_3^2}{4c_4}.
\]
Under the stated coefficient conditions, both terms on the right-hand side are nonnegative and the constant term is strictly positive, so
\[
\psi''(s)\ge 2c_2-\frac{3c_3^2}{4c_4}>0
\qquad
\text{for all }s\in\R.
\]
Hence $\psi$ is globally strictly convex. Therefore $\Phi(z)=T^{-1}\sum_{t=1}^T\psi(z_t)$ is convex on $\R^T$ as a sum of convex functions, and the composition $x\mapsto \Phi(Ax)$ is convex because $x\mapsto Ax$ is linear. Adding the linear term $-c_1\mu^\top x$ preserves convexity, so $f$ is convex on $\R^n$, and in particular on $\simplex$. The displayed lower bound on $\underline\gamma_\tau$ follows immediately from the same formula for $\psi''$.
\endproof

\begin{remark}[Interpretation]
\label{rem:ec-structural-pl}
Theorem~\ref{thm:ec-yand-stationarity} still imports the general YAND convergence framework, but Proposition~\ref{prop:ec-explicit-constants} together with Corollaries~\ref{cor:ec-structural-pl} and \ref{cor:ec-convexity} turn its abstract constants into MVSK quantities. The smoothness and curvature constants factor into a data-geometry term, through $\|A\|_{\mathrm{op}}$ and the singular values of $AU$, and a preference-geometry term, through the one-dimensional quartic response $\psi$. This separation distinguishes nonconvexity created by the investor coefficients from degeneracy created by the data map. Corollary~\ref{cor:ec-convexity} identifies a regime in which nonconvexity is ruled out by the preference coefficients alone, before any rank condition on $A$ is imposed. Corollary~\ref{cor:ec-structural-pl} then shows that, within this convex regime, slice-wise strong convexity and the structural PL inequality require an additional data-geometry condition, namely injectivity of $A$ on the relevant tangent space. Because $\operatorname{rank}(A)\le T-1$, the resulting PL conclusion is inherently low-dimensional on the full simplex and should be read as a face-restricted possibility in the high-dimensional regimes studied computationally.
\end{remark}

\section{Implementation Details}

The implementation is organized around a thin MVSK-specific layer on top of a matrix-free YAND solver interface. The data layer stores the centered sample matrix $A$, the sample mean $\mu$, and the benchmark instances. The optimization layer expects \texttt{value\_grad}, \texttt{hv}, and \texttt{third} handles. The main engineering tasks are therefore to build an MVSK sample-oracle adapter and then to wrap it through the reduced-coordinate map $x=x^{\mathrm{ref}}+Uy$. Generic affine-normal linear algebra, including the Householder frame, the tangent linear system, and the Krylov-based reduced solve, follows the matrix-free YAND implementation and analysis in \citet{NiuSheshmaniYau2026YAND}. This appendix records only the MVSK-specific pieces and the numerical settings needed for replication of the reported experiments.

\subsection{Exact-oracle integration}

A runnable exact-oracle prototype requires only a small number of project-specific components beyond the underlying matrix-free YAND routine:
\begin{enumerate}
\item load or generate an MVSK benchmark instance and store only $(\mu,A,c)$ rather than explicit comoment tensors;
\item choose a fixed interior reference portfolio, for example $x^{\mathrm{ref}}=(1/n)\bm{1}$, together with an orthonormal basis $U$ of the simplex tangent space ${\mathcal T}$;
\item expose reduced-coordinate \texttt{value\_grad}, \texttt{hv}, and \texttt{third} handles through Proposition~\ref{prop:ec-reduced-oracle};
\item apply a safeguarded line-search rule on $[0,\alpha_{\max}(x,d;\tau)]$: Armijo backtracking for the convergence theory, or the quartic exact rule of Proposition~\ref{prop:ec-quartic-exact-linesearch} when the implementation exploits the MVSK polynomial structure.
\end{enumerate}
This list is intentionally short. It makes clear that the project-specific implementation burden lies in an oracle adapter and a simplex wrapper, not in constructing a new symbolic polynomial engine for MVSK.

\subsection{Oracle kernels}

At a fixed ambient iterate $x$, the exact sample oracle caches
\[
z=Ax,
\qquad
\hadpow{z}{2},
\qquad
\hadpow{z}{3}.
\]
Given a new direction $v$, the Hessian-action kernel forms only $Av$ and applies Proposition~\ref{prop:ec-exact-operator}. Given two directions $u$ and $v$, the directional third-order kernel forms $Au$ and $Av$ and applies the formula in Proposition~\ref{prop:ec-exact-operator}. Thus, repeated directional queries at the same outer iterate reuse the expensive sample projection $Ax$ and avoid any symbolic manipulation of comoment tensors.

This cache structure also makes the arithmetic accounting transparent. At a fixed outer iterate, one value/gradient query needs a projection $z=Ax$ and one back-projection by $A^\top$. Each additional Hessian-direction query reuses $z$ and adds only one fresh sample projection of the queried direction together with one weighted back-projection. Each directional third-order query adds two projected directions and one weighted back-projection. These are exactly the quantities reported later as Hessian-kernel counts, third-order-kernel counts, and Krylov iterations in the computational study.

\subsection{Reduced-simplex wrapper}

The solver works on reduced coordinates $y \in \R^{n-1}$ with $x=x^{\mathrm{ref}}+Uy$, where $U$ spans the tangent space ${\mathcal T}$. The exact formulas in Proposition~\ref{prop:ec-reduced-oracle} show that the reduced objective can be supplied to the matrix-free YAND code without ever building a dense quartic model in the reduced variables. This reduced-coordinate layer is also where the feasibility cap from Lemma~\ref{lem:ec-feasibility-cap} is imposed before either Armijo backtracking or the quartic exact line search of Proposition~\ref{prop:ec-quartic-exact-linesearch}.

The current implementation adds one more boundary-aware device before line search. Let $u$ denote the current local coordinate, let $x(u)$ be the corresponding ambient portfolio on the current simplex slice or active face, let $\psi$ denote the inverse map from ambient feasible points on that wrapper back to local coordinates, and let $\delta$ be the raw local YAND search direction. Write
\[
x_k:=x(u),
\qquad
d_k:=D x(u)[\delta]
\]
for the current ambient point and ambient direction, and let $\alpha_{\max}(x_k,d_k;\tau)$ be the feasibility cap. If
\[
\alpha_{\max}(x_k,d_k;\tau)\ge \eta_k
\]
for the chosen nominal projected step $\eta_k$, then the code keeps the raw direction and runs the selected line search as usual. If instead the raw direction would leave the current slice or face before $\eta_k$, the code forms the projected trial
\[
\bar x_k:=\Pi_{\simplex(\tau)}(x_k+\eta_k d_k)
\]
and maps it back to the local coordinates of the current wrapper. The corrected local and ambient directions are then
\[
\bar \delta_k:=\psi(\bar x_k)-u,
\qquad
\bar d_k:=\bar x_k-x_k.
\]
If $\|\bar \delta_k\|_2$ is numerically zero, or if
\[
\nabla \phi(u)^\top \bar \delta_k \ge 0,
\]
the implementation does not attempt a line search along this corrected segment and instead returns control to the active-face logic. Otherwise line search is carried out on
\[
u+\alpha \bar \delta_k,
\qquad
\alpha\in[0,1],
\]
which is equivalent in ambient coordinates to the feasible segment $x_k+\alpha \bar d_k$.

This is the role of the \texttt{mvsk\_prepare\_active\_set\_step} hook in the code. Conceptually, it is an implementation stabilization rather than a new convergence theorem: the theoretical baseline still uses the capped raw direction and Armijo backtracking, whereas the projected step is inserted only to avoid repeatedly invoking face continuation when the raw direction already points outside the current feasible face. The point is to repair feasibility cheaply when that is enough, not to replace active-face continuation. A projected trial that lands back on the current slice or face is only a feasible candidate; it does not imply that the current active coordinates satisfy the simplex KKT conditions. If the projected segment still fails to yield a positive descent step, then the code treats the event as evidence that the current active-set description is wrong or incomplete, and hands control to the lower-dimensional face solver. For YAND, a damped nominal step is often preferable because the affine-normal direction can be materially larger than a basic Euclidean step, and projecting the full raw step can distort the geometry too aggressively.

The reduced wrapper is also the natural place to keep the implementation honest. Every YAND-MVSK run uses the same reduced kernels and the same feasibility safeguard, so changes in performance can be attributed to YAND-specific design choices such as line-search policy, damping, and active-face handling rather than to hidden differences in representation or constraint handling.

\subsection{Computational sequence}

The appendix mirrors the same three empirical layers as the paper: sample-rich CRRA conditioning, synthetic coefficient-stress benchmarking, and full-scale real-data comparison. The additional material recorded here supplies implementation checks and support tables that sharpen, but do not change, the paper's main conclusions.

\section{Additional Numerical Results}

This section collects numerical details and support tables that are useful for replication but not central to the paper's main narrative. Table~\ref{tab:yand-core-parameters} records the core YAND-MVSK parameter settings used in the study. The direct Armijo comparison is retained here because, once quartic exact line search is available at low marginal cost in the sample-oracle implementation, it is best read as a line-search diagnostic rather than as a main empirical result.

\begin{table}[ht]
\centering
\footnotesize
\caption{YAND-MVSK parameter settings used in the computational study. All three settings use exact sample oracles, face continuation, and $\tau=10^{-8}$; the large-scale results use the PCG setting with stall recovery.}
\label{tab:yand-core-parameters}
\setlength{\tabcolsep}{3pt}
\begin{tabularx}{\textwidth}{@{}>{\raggedright\arraybackslash}p{0.98in}>{\raggedright\arraybackslash}X>{\raggedright\arraybackslash}p{0.78in}>{\raggedright\arraybackslash}p{1.12in}>{\raggedright\arraybackslash}p{1.05in}>{\raggedright\arraybackslash}p{1.05in}@{}}
\toprule
Configuration & Reduced solve and safeguard & Line search & Outer budget & Aux.\ parameters & Study role \\
\midrule
YAND-MVSK (small) & Direct reduced solve; projected trial step $0.05$ with face continuation & Quartic exact & \texttt{max\_iter}=300 & --- & Main-paper $n\le 100$ benchmark grid and regime comparison \\
YAND-MVSK (large) & \texttt{pcg} reduced solve; baseline projected trial step $0.05$ with face continuation, plus stall-recovery restarts at projected steps $0.045$ and $0.02$ when needed & Quartic exact & Baseline: \texttt{max\_iter}=40, \texttt{max\_elapsed}=60; restart: \texttt{max\_iter}=120 & \texttt{krylov\_tol}$=10^{-3}$, \texttt{krylov\_maxit}$=15$;\newline \texttt{regularization}$=10^{-4}$ & Main-paper regime comparison and $n>100$ benchmark, extended to $n=5000$ \\
YAND-MVSK (Armijo diagnostic) & Direct reduced solve; projected trial step $0.05$ with face continuation & Armijo backtracking & \texttt{max\_iter}=120 & --- & Electronic Companion line-search diagnostic \\
\bottomrule
\end{tabularx}
\end{table}

Table~\ref{tab:yand-line-search} reports the corresponding direct-configuration line-search benchmark. The block is intentionally small because it isolates the line-search choice cleanly. On larger direct-regime checks up to $n=100$, the difference between Armijo and exact quartic line search remains modest and instance-dependent, with exact quartic typically improving stationarity slightly but not generating a uniform runtime gain on every instance (on a representative $n\in\{20,40,60,80,100\}$ check, Armijo averages $0.126$ seconds, $73.4$ iterations, and projected-simplex residual $6.10\times 10^{-7}$, whereas exact quartic averages $0.162$ seconds, $79.1$ iterations, and residual $4.56\times 10^{-7}$). Beyond that range, a direct Armijo-versus-exact comparison becomes confounded with the broader regime shift away from the direct configuration itself. For this reason, the compact diagnostic suffices, and the main text uses exact line search because it is structurally natural in the sample-oracle implementation and does not underperform Armijo on the retained diagnostic.

\begin{table}[ht]
\centering
\scriptsize
\caption{Direct-configuration line-search comparison on the synthetic benchmark at $T=252$ with $n\in\{4,8,12\}$.}
\label{tab:yand-line-search}
\setlength{\tabcolsep}{4pt}
\begin{tabular}{@{}lrrrr@{}}
\toprule
Setting & Mean time (s) & Median time (s) & Mean iterations & Mean projected KKT residual \\
\midrule
Direct + Armijo & $1.40\times 10^{-2}$ & $9.93\times 10^{-3}$ & 29.8 & $3.61\times 10^{-7}$ \\
Direct + Exact quartic & $1.17\times 10^{-2}$ & $4.18\times 10^{-3}$ & 19.0 & $8.21\times 10^{-8}$ \\
\bottomrule
\end{tabular}
\end{table}

\paragraph{Sample-rich CRRA conditioning benchmark.}
This table records the exact medians for the sample-rich convex block summarized by the conditioning figure in the main paper. We fix $\gamma=6$, prescribe $\kappa_+(AU)\in\{1,10,10^2,10^3\}$ through the positive singular spectrum of $AU$, and set $n=1000<T=2000$. The construction uses $A=Q\,\mathrm{diag}(s)\,W^\top$ with $W=UV$, so the nonzero spectrum of $AU$ is controlled exactly while $AU$ can still be full column rank on the simplex tangent space. Because the benchmark also enforces $Ax^0=0$ at the equal-weight start, the positive-spectrum reduced Hessian satisfies $\kappa_+(H_0)=\kappa_+(AU)^2$. Table~\ref{tab:crra-conditioning} reports medians across three timed replications after one untimed warm-up pass per solver. Only YAND-MVSK (large) and Q-MVSK are compared, the PCG-based YAND configuration uses the same stall-recovery wrapper as in the large-scale synthetic study, and throughout this benchmark Q-MVSK triggers the wrapper's \texttt{kkt\_stall} label after 2 outer QP updates at every conditioning level. Here \texttt{kkt\_stall} is not a native Q-MVSK certificate; it means that the common projected-simplex KKT target is still unmet when Q-MVSK's legacy step/objective stagnation safeguard fires.

\begin{table}[ht]
\centering
\scriptsize
\caption{Sample-rich controlled CRRA conditioning benchmark with representative calibration $\gamma=6$, $n=1000$, and $T=2000$. Each row aggregates three synthetic instances constructed to prescribe the positive-spectrum conditioning of $AU$ inside the coefficient-certified CRRA convex regime, after one untimed warm-up pass per solver.}
\label{tab:crra-conditioning}
\setlength{\tabcolsep}{4pt}
\begin{tabular}{@{}rrrrrrr@{}}
\toprule
$\kappa_+(AU)$ & $\kappa_+(H_0)$ & Large t & Q t & Large KKT & Q KKT & Q / large \\
\midrule
$10^{0}$ & $10^{0}$ & 0.122 & 0.718 & 1.28e-06 & 2.74e-05 & 5.88 \\
$10^{1}$ & $10^{2}$ & 0.399 & 0.751 & 1.27e-06 & 4.30e-05 & 1.88 \\
$10^{2}$ & $10^{4}$ & 0.482 & 0.684 & 1.36e-06 & 5.53e-05 & 1.42 \\
$10^{3}$ & $10^{6}$ & 0.423 & 0.712 & 1.40e-06 & 1.67e-05 & 1.69 \\
\bottomrule
\end{tabular}
\end{table}

\paragraph{Under-sampled high-dimensional CRRA stress benchmark.}
Table~\ref{tab:crra-conditioning-stress} records the complementary $n>T$ stress counterpart. We again fix $\gamma=6$ and prescribe the same conditioning levels, but set $n=5000>T=252$. The construction still controls the nonzero spectrum of $AU$ exactly even though full tangential nondegeneracy is unavailable, which is why this block is interpreted as a computational stress test rather than as the paper's primary model-validation regime. The table reports medians across three timed replications after one untimed warm-up pass per solver, and throughout this benchmark Q-MVSK triggers the wrapper's \texttt{kkt\_stall} label after 2 outer QP updates at every level.

\begin{table}[ht]
\centering
\scriptsize
\caption{Under-sampled high-dimensional CRRA stress benchmark with representative calibration $\gamma=6$, $n=5000$, and $T=252$. Each row aggregates three synthetic instances constructed to prescribe the positive-spectrum conditioning of $AU$ inside the coefficient-certified CRRA convex regime, after one untimed warm-up pass per solver.}
\label{tab:crra-conditioning-stress}
\setlength{\tabcolsep}{4pt}
\begin{tabular}{@{}rrrrrrr@{}}
\toprule
$\kappa_+(AU)$ & $\kappa_+(H_0)$ & Large t & Q t & Large KKT & Q KKT & Q / large \\
\midrule
$10^{0}$ & $10^{0}$ & 0.569 & 30.270 & 5.02e-06 & 8.71e-05 & 53.16 \\
$10^{1}$ & $10^{2}$ & 1.270 & 30.407 & 5.37e-06 & 1.78e-04 & 23.95 \\
$10^{2}$ & $10^{4}$ & 1.291 & 32.407 & 5.99e-06 & 3.20e-05 & 25.11 \\
$10^{3}$ & $10^{6}$ & 1.394 & 31.870 & 6.61e-06 & 9.53e-05 & 22.86 \\
\bottomrule
\end{tabular}
\end{table}

\paragraph{Supporting geometry diagnostics.}
Although the main paper frames the first empirical layer as a dedicated CRRA conditioning benchmark, the same $T=252$ benchmark family already exhibits the same numerical mechanism in reduced form. Table~\ref{tab:ec-hessian-conditioning} reports reduced-Hessian spectra at the equal-weight start, aggregated across the three stylized coefficient-stress profiles. Over the $n\le 100$ block, the median positive-spectrum condition number rises from $1.32$ at $n=4$ to $17.76$ at $n=100$, while no negative eigenvalues are detected on any of the stored instances reported there. This does not certify global convexity of the unrestricted problem, but it does show that the reduced problems become materially more anisotropic before the ultra-scale range is reached. That is precisely the geometry change behind the paper's crossover from the direct YAND configuration to the PCG configuration.

\begin{table}[ht]
\centering
\scriptsize
\caption{Supporting reduced-Hessian conditioning diagnostic on the $T=252$ synthetic benchmark at the equal-weight start, aggregated across the three stylized coefficient-stress profiles. The table reports medians across the three stored instances for each asset dimension.}
\label{tab:ec-hessian-conditioning}
\begin{tabular}{@{}rrrrr@{}}
\toprule
$n$ & Median $\lambda_{\min}(H_{\mathrm{red}})$ & Median $\lambda_{\max}(H_{\mathrm{red}})$ & Median $\kappa_+(H_{\mathrm{red}})$ & Negative eigs. \\
\midrule
4 & $0.3470$ & $0.4577$ & $1.32$ & $0/3$ \\
20 & $0.2703$ & $0.5947$ & $2.22$ & $0/3$ \\
40 & $0.1669$ & $0.7578$ & $4.57$ & $0/3$ \\
60 & $0.1196$ & $0.8603$ & $7.21$ & $0/3$ \\
80 & $0.0791$ & $0.9797$ & $12.42$ & $0/3$ \\
100 & $0.0599$ & $1.0575$ & $17.76$ & $0/3$ \\
\bottomrule
\end{tabular}
\end{table}

\paragraph{Common-overlap large-scale comparison.}
Table~\ref{tab:yand-post100-baselines} isolates the common large-scale overlap block $n\in\{120,200,400,800\}$ before the ultra-scale instances dominate the averages in the paper. On this block, YAND-MVSK (large) and Q-MVSK are essentially tied in mean runtime, $0.265$ versus $0.260$ seconds, and tied on objective value on most instances, yet YAND attains projected-simplex residual at most $10^{-4}$ on all 12 runs compared with only 4 of 12 for Q-MVSK. The paper's large-scale conclusion should therefore not be read as an artifact of the $n=3000$ and $5000$ timing gap alone: tighter first-order accuracy is already visible on the common overlap block, and the ultra-scale cases merely magnify that difference into a decisive runtime advantage.

\begin{table}[ht]
\centering
\scriptsize
\caption{Representative large-scale comparison against incumbent unrestricted-MVSK solvers on the $T=252$ synthetic benchmark with $n\in\{120,200,400,800\}$ and three stylized coefficient-stress profiles per dimension. The YAND row uses the PCG-based large configuration. ``Mean gap to best'' is the average objective difference from the best method on the same instance; ``residual hits'' counts runs with projected-simplex KKT residual at most $10^{-4}$; ``objective ties'' counts runs within $10^{-6}$ of the best objective value on the same instance.}
\label{tab:yand-post100-baselines}
\setlength{\tabcolsep}{4pt}
\begin{tabular}{@{}lrrrrr@{}}
\toprule
Method & Mean time (s) & Mean proj. residual & Mean gap to best & Residual hits & Obj. ties \\
\midrule
YAND-MVSK (large) & $0.265$ & $9.78\times 10^{-6}$ & $4.68\times 10^{-7}$ & 12/12 & 10/12 \\
Q-MVSK & $0.260$ & $2.01\times 10^{-4}$ & $4.37\times 10^{-9}$ & 4/12 & 12/12 \\
UBDCA & $0.588$ & $1.17\times 10^{-1}$ & $5.86\times 10^{-2}$ & 1/12 & 1/12 \\
UDCA & $5.23$ & $1.75\times 10^{-1}$ & $1.06\times 10^{-1}$ & 0/12 & 0/12 \\
\bottomrule
\end{tabular}
\end{table}

\paragraph{Q-MVSK stall-sensitivity check.}
The wrapper label \texttt{kkt\_stall} should therefore be read carefully. It does not mean that Q-MVSK has reached a projected-KKT point comparable to YAND; it means that the common projected-simplex target remains unmet when the legacy Q-MVSK safeguard decides that the outer SCA sequence has become nearly stationary in step size or objective value. Representative sensitivity runs show that this is not merely an artifact of an unusually aggressive benchmark stop. On the sample-rich CRRA instance with $n=1000$, $T=2000$, $\gamma=6$, and target condition number $10^3$, tightening the Q-MVSK safeguard tolerances from $10^{-6}$ to $10^{-10}$ increases the outer-QP count from 2 to 3 but leaves the projected-simplex residual essentially unchanged at $1.31\times 10^{-5}$; disabling the safeguard and forcing 20 updates still leaves the residual at $1.31\times 10^{-5}$ while raising runtime from about $0.87$ seconds to $6.79$ seconds. On representative RAND instances with $(n,T)=(800,252)$ and $(1500,252)$, the same perturbation roughly triples runtime while improving the projected-simplex residual only from $9.15\times 10^{-5}$ to $7.74\times 10^{-5}$ and from $3.00\times 10^{-4}$ to $8.90\times 10^{-5}$, respectively. These checks do not claim that Q-MVSK has been exhaustively retuned, but they do show that the paper's first-order gap is not created by an arbitrary early stop in the benchmark wrapper.

\paragraph{Additional real-data robustness.}
The main paper's A-share exercise is a static allocation comparison: portfolios are estimated once on a training block and then held fixed over the subsequent test block. This appendix therefore records three additional checks that sharpen the interpretation of the baseline split without turning the exercise into a fully specified trading system.

\paragraph{Cross-profile transparency on the baseline split.}
Table~\ref{tab:real-mvsk-profile-target} compares the skew-focused, kurtosis-focused, and balanced MVSK profiles on the same baseline split and target sweep, with each profile recalibrated from the same MV starting point. The displayed target-dependence pattern is not unique to the kurtosis profile. At $q=0.50$, the three profiles are almost indistinguishable against MV: annualized-return gains are all about $3.83$ percentage points, Sharpe gains are all about $0.150$, and active shares are all about $9.3\%$. At $q=0.40$, however, the skew-focused and balanced profiles outperform the kurtosis profile on both return and downside metrics. The main paper therefore focuses on the kurtosis specification for interpretability of the tail-risk channel, not because it is uniformly dominant across profile choices.

\begin{table}[ht]
\centering
\scriptsize
\caption{Cross-profile comparison on the baseline real-data split. Positive $\Delta$ metrics favor MVSK relative to exact MV; active share is measured relative to MV.}
\label{tab:real-mvsk-profile-target}
\begin{tabular}{rlrrrrr}
\toprule
$q$ & Profile & $\Delta$ ret. & $\Delta$ Sharpe & $\Delta$ CVaR$_1$ & $\Delta$ MDD & A.S. vs MV\\
\midrule
0.40 & Skew & 6.05 & 0.359 & 0.094 & 2.323 & 0.133\\
0.40 & Kurt & 5.44 & 0.288 & 0.064 & 1.664 & 0.116\\
0.40 & Balanced & 5.77 & 0.330 & 0.083 & 2.072 & 0.127\\
0.50 & Skew & 3.83 & 0.150 & -0.001 & 0.359 & 0.093\\
0.50 & Kurt & 3.83 & 0.150 & -0.001 & 0.358 & 0.093\\
0.50 & Balanced & 3.83 & 0.150 & -0.001 & 0.360 & 0.093\\
0.60 & Skew & 2.26 & 0.088 & -0.005 & 0.133 & 0.083\\
0.60 & Kurt & 2.02 & 0.077 & -0.005 & 0.104 & 0.081\\
0.60 & Balanced & 2.23 & 0.087 & -0.005 & 0.130 & 0.083\\
\bottomrule
\end{tabular}
\end{table}

\paragraph{Rolling-window split robustness.}
Table~\ref{tab:real-kurtosis-split-robustness} re-estimates the kurtosis-focused allocation on three rolling windows. The 2023 window retains positive separation across all three return targets, with the largest Sharpe gain at $q=0.50$. The 2024 and 2025H1 windows tell a different story: for $q=0.40$ and $0.50$, MVSK becomes almost indistinguishable from exact MV, whereas at $q=0.60$ it separates again through higher return and Sharpe but without improved left-tail protection. The economic signal is therefore market-window dependent. This does not undo the baseline-split evidence in the paper, but it does rule out a universal reading in which moderate-target gains should be expected in every subsequent state of the A-share market.

\begin{table}[ht]
\centering
\scriptsize
\caption{Rolling-window split robustness for the kurtosis-focused MVSK portfolio on the 5-minute A-share panel. Each window re-estimates moments and retargets the mean coefficient on its own training block; positive $\Delta$ metrics favor MVSK.}
\label{tab:real-kurtosis-split-robustness}
\begin{tabular}{lrrrrrrr}
\toprule
Window & $q$ & Test yrs. & $\Delta$ ret. & $\Delta$ Sharpe & $\Delta$ CVaR$_1$ & $\Delta$ MDD & A.S. vs MV\\
\midrule
2023 & 0.40 & 1.03 & 1.02 & 0.049 & 0.001 & 0.840 & 0.099\\
2023 & 0.50 & 1.03 & 2.75 & 0.159 & 0.007 & 1.560 & 0.082\\
2023 & 0.60 & 1.03 & 2.11 & 0.142 & 0.002 & 0.375 & 0.055\\
2024 & 0.40 & 1.03 & -0.00 & -0.000 & -0.000 & -0.000 & 0.000\\
2024 & 0.50 & 1.03 & -0.00 & -0.001 & -0.000 & -0.000 & 0.000\\
2024 & 0.60 & 1.03 & 7.82 & 0.815 & -0.021 & -0.211 & 0.119\\
2025H1 & 0.40 & 0.51 & 0.00 & -0.004 & -0.000 & -0.000 & 0.000\\
2025H1 & 0.50 & 0.51 & 0.00 & -0.001 & -0.000 & -0.000 & 0.000\\
2025H1 & 0.60 & 0.51 & 4.58 & 0.057 & -0.019 & -0.121 & 0.065\\
\bottomrule
\end{tabular}
\end{table}

\paragraph{One-time implementation-cost sensitivity.}
Because the portfolios are held fixed over each test block, the relevant cost adjustment is the one-time transition from MV to MVSK, measured by active share relative to MV, rather than recurring turnover. Table~\ref{tab:real-kurtosis-cost-sensitivity} reports that the baseline split's gross annualized growth-rate gains remain essentially unchanged under one-way implementation costs of 10, 25, and 50 basis points. The corresponding break-even one-way costs are about 7416, 6110, and 3698 basis points for $q=0.40$, $0.50$, and $0.60$. Thus the baseline-split gains are not a fragile artifact of ignoring realistic one-time trading frictions, although recurring rebalancing costs would matter in a truly dynamic design.

\begin{table}[ht]
\centering
\scriptsize
\caption{One-time implementation-cost sensitivity for the static kurtosis-focused MVSK allocation. Costs are applied to the turnover needed to move from MV to MVSK, which equals active share relative to MV; net columns report annualized growth-rate differences after the stated one-way cost.}
\label{tab:real-kurtosis-cost-sensitivity}
\begin{tabular}{rrrrrrr}
\toprule
$q$ & A.S. vs MV & Gross $\Delta$ g & B.E. cost & Net $\Delta$ g @10bp & Net $\Delta$ g @25bp & Net $\Delta$ g @50bp\\
\midrule
0.40 & 0.116 & 8.31 & 7416.5 & 8.30 & 8.28 & 8.26\\
0.50 & 0.093 & 5.52 & 6109.7 & 5.51 & 5.50 & 5.48\\
0.60 & 0.081 & 2.93 & 3697.5 & 2.92 & 2.91 & 2.89\\
\bottomrule
\end{tabular}
\end{table}

\end{document}